\documentclass[journal,12pt,onecolumn,draftclsnofoot,]{IEEEtran}
\usepackage[T1]{fontenc}
\usepackage{tikz}
\usepackage{etoolbox}
\usepackage{enumitem}
\newcommand*{\circled}[1]{\lower.7ex\hbox{\tikz\draw (0pt, 0pt)%
    circle (.5em) node {\makebox[1em][c]{\small #1}};}}
\robustify{\circled}
\usepackage{amsmath}
\usepackage{amsthm}
\usepackage{amssymb} 
\usepackage{hyperref}
\usepackage{mathrsfs}
\usepackage{bbm}
\usepackage{amsfonts}
\usepackage{graphicx,cite,epsfig,amssymb,amsmath}
\usepackage{color,xcolor}
\usepackage{bm}
\definecolor{red}{rgb}{1.00, 0.00, 0.00}  
\usepackage{pifont}
\usepackage{stmaryrd}
\usepackage{setspace}
\usepackage{subfigure}
\usepackage{cite}
\usepackage{array}
\usepackage{float}
\usepackage{multirow}
\usepackage{algorithm,algpseudocode}
\usepackage{epstopdf}
\usepackage{mathtools}
\usepackage{amssymb}
\usepackage{diagbox}
\usepackage{booktabs}
\usepackage{makecell}

\usepackage{makecell}
\usepackage{xcolor}
\usepackage{pifont}
\makeatletter

\newtheoremstyle{examplestyle}
  {}
  {}
  {}
  {}
  {\itshape}
  {}
  {.5em}
  {\thmname{#1}\thmnumber{ #2}\thmnote{ (#3)}}
\theoremstyle{examplestyle}
\newtheorem{example}{Example}

\newcommand{\Rmnum}[1]{\expandafter\@slowromancap\romannumeral #1@}

\makeatother
\usepackage{pifont} 
\usepackage{epstopdf}
\usepackage{mathtools}
\usepackage{diagbox}
\usepackage{booktabs}
\usepackage{multirow}
\usepackage{booktabs}
\usepackage{makecell}
\usepackage{cancel}
\usepackage{float} 
\usepackage{subfigure} 
\usepackage{colortbl}

\usepackage{setspace}
\definecolor{SP}{RGB}{30,160,156}

\newtheorem{lemma}{Lemma}
\usepackage{setspace}
\usepackage{tikz}
\renewenvironment{proof}{{\itshape Proof.}}{}

\begin{document}
\normalsize
\title{What If We Had Used a Different App? Reliable Counterfactual KPI Analysis in Wireless Systems}

\author{Qiushuo Hou, Sangwoo Park, Matteo Zecchin, Yunlong Cai, Guanding Yu, and Osvaldo Simeone
\thanks{

Q. Hou, Y. Cai, and G. Yu are with the College of Information Science and Electronic Engineering, Zhejiang University, Hangzhou 310027, China (e-mail: \{qshou, ylcai, yuguanding\}@zju.edu.cn).

Sangwoo Park, Matteo Zecchin, and Osvaldo Simeone are with the King’s Communications, Learning \& Information Processing (KCLIP) lab within the Centre for Intelligent Information Processing Systems (CIIPS), Department of Engineering, King’s College London, London WC2R 2LS, U.K. (e-mail: \{sangwoo.park, matteo.1.zecchin, osvaldo.simeone\}@kcl.ac.uk). 

The work of M. Zecchin and O. Simeone is supported by the European Union’s Horizon Europe project CENTRIC (101096379). The work of O. Simeone is also supported by an Open Fellowship of the EPSRC (EP/W024101/1), and by the EPSRC project (EP/X011852/1). 
}}

\maketitle
\vspace{-1.5cm}
\begin{abstract}
In modern wireless network architectures, such as Open Radio Access Network (O-RAN), the operation of the radio access network (RAN) is managed by applications, or apps for short, deployed at intelligent controllers. These apps  are selected from a given catalog based on current contextual information. For instance, a scheduling app may be selected on the basis of current traffic and network conditions. Once an app is chosen and run, it is no longer possible to directly test the key performance indicators (KPIs) that would have been obtained with another app. In other words, we can never simultaneously observe both the actual KPI, obtained by the selected app, and the counterfactual KPI, which would have been attained with another app, for the same network condition, making individual-level counterfactual KPIs analysis particularly challenging. This what-if analysis, however, would be valuable to monitor and optimize the network operation, e.g., to identify suboptimal app selection strategies. This paper addresses the problem of estimating the values of KPIs that would have been obtained if a different app had  been implemented by the RAN. To this end, we propose a conformal-prediction-based counterfactual analysis method for wireless systems that provides reliable error bars for the estimated KPIs, despite the inherent covariate shift between logged and test data. Experimental results for medium access control-layer apps and for physical-layer apps demonstrate the merits of the proposed method.
\end{abstract}
\section{Introduction}

\subsection{Context and Motivation}

\begin{figure*}[htbp]
    \centering
    \includegraphics[width=9cm]{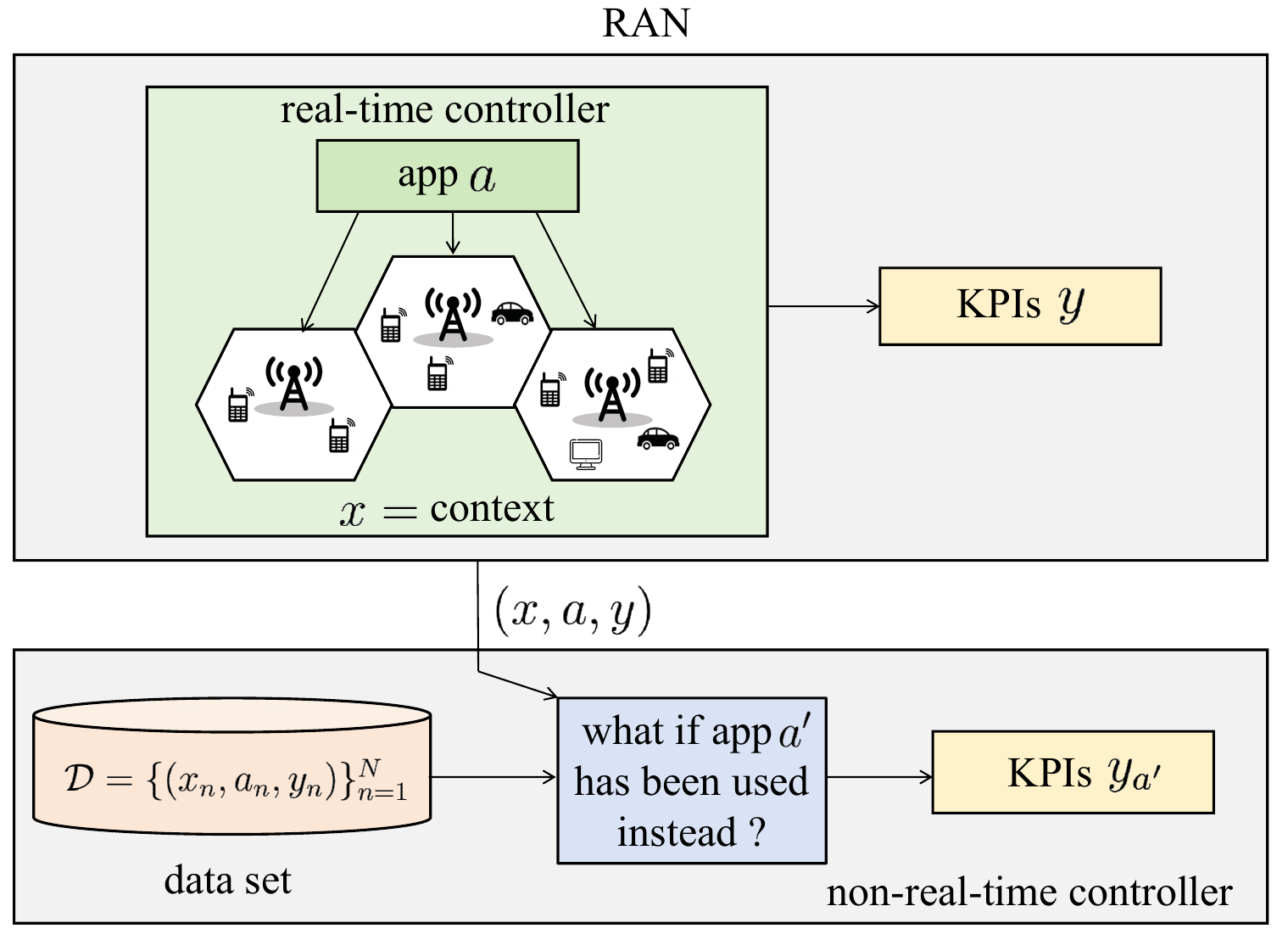}
    \vspace{-0.2cm}
    \caption{In the wireless system under study, the radio access network (RAN) is managed by a non-real-time controller. The controller collects data from the RAN about key performance indicators (KPIs) attained by apps implemented by the RAN. Accordingly, the controller logs data in the form $(x,a,y)$ as the data set, where $x$ is the context, $a$ is the app identifier, and $y$ is the KPIs. The controller implements the counterfactual analysis by answering a \emph{what-if} question: Given that app $a$ has obtained KPIs $y$ for context $x$, what would the KPIs have been for the same context $x$ had some other app $a'\neq a$ been selected by the non-real-time controller?}
    \label{fig:O-RAN}
    \vspace{-0.4cm}
\end{figure*}

As illustrated in Fig. \ref{fig:O-RAN}, in modern wireless systems such as Open Radio Access Network (O-RAN), the operation of the RAN is managed by algorithms, or apps, that are selected based on current contextual information by a non-real-time controller in the cloud  \cite{O-RAN1,O-RAN2}. For example, based on the current backlogs of a set of users and on the respective channel quality indicators, the controller may choose different scheduling apps that strike a desirable balance between throughput and fairness \cite{scheduling1, scheduling2} (see Fig. \ref{fig:system_model}).  Once an app is chosen and run, it is no longer possible to directly test the performance that would have been obtained with a different app. This what-if analysis, however, would be valuable to monitor and optimize the network operation, e.g., to identify suboptimal app selection strategies \cite{COMA, DT-counterfactual-wireless2}. Furthermore, counterfactual analyses can be instrumental in providing \emph{explanations} about the decisions made by the controller \cite{XAI-counterfactual1, XAI-counterfactual2, XAI-counterfactual3}. 

As an example, for the setting in Fig. \ref{fig:system_model}, assume that a given scheduling algorithm, such as proportional fair channel aware (PFCA) \cite{mobile}, was selected by the controller on the basis of the available contextual information given by the initial backlogs and the channel quality indicators. What would the final backlogs have been if the controller had chosen a different scheduling algorithm, such as round-robin (RR)?

Counterfactual analysis is the highest layer in the causal inference hierarchy\cite{causal in wireless}. The lowest layer, \emph{association}, focuses on correlations, while the intermediate layer, \emph{intervention}, involves actual experimentation on the environment, as in reinforcement learning. Counterfactual analysis represents the most sophisticated form of causal inference, answering “what-if” questions about past events using data, without requiring interventions in the environment.

The three layers form a strict hierarchy in which each layer subsumes the capabilities of the ones below, providing progressively more powerful tools for understanding and reasoning about causality in complex systems. Specifically, association lacks sufficient understanding of causality, i.e., cannot distinguish between correlation and causality. Intervention faces practical limitations in real-world applications. For example, in wireless networks, network-wide interventions could disrupt service quality and user experience. In contrast, counterfactual analysis supports reasoning about alternative scenarios without actual interventions, making it particularly valuable for wireless systems analysis and optimization.

This paper addresses the what-if problem of estimating the values of key performance indicators (KPIs) that would have been obtained if a different app had been implemented by the RAN. As illustrated in Fig.  \ref{fig:O-RAN},  modern networks log data about the measured KPIs for given apps and contexts for diagnostics purposes and optimization \cite{SCOPE}. Thus, the question of interest is:  At run time, given that a deployed app has produced certain KPI levels, can the logged data be leveraged to estimate the KPIs that would have been obtained if the controller had chosen a different app?

\begin{figure*}[htp]
    \centering
    \subfigure[]{
    \includegraphics[width=6.5cm]{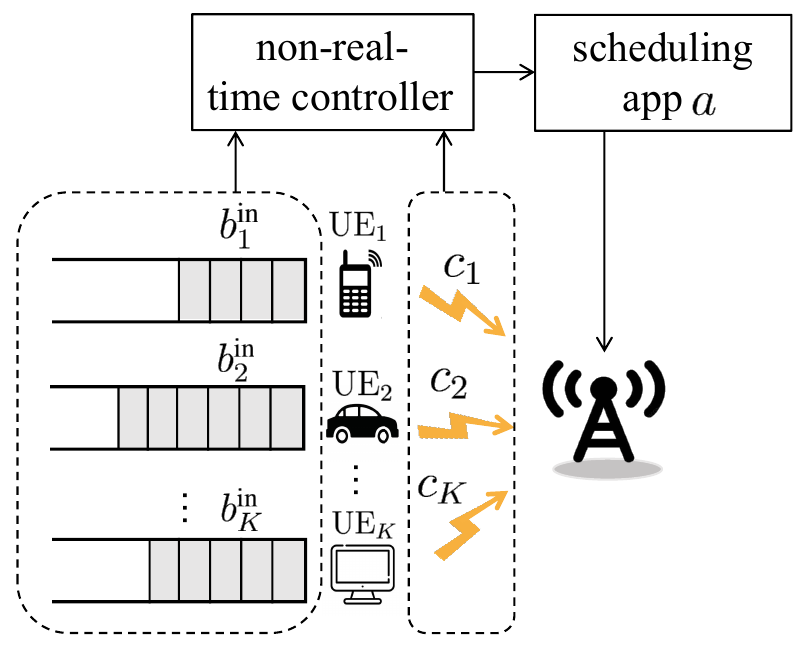}
    \vspace{-0.5cm}
    
    \label{fig:model_a}
    }
    \subfigure[]
    {
    \centering
    \includegraphics[width=7cm]{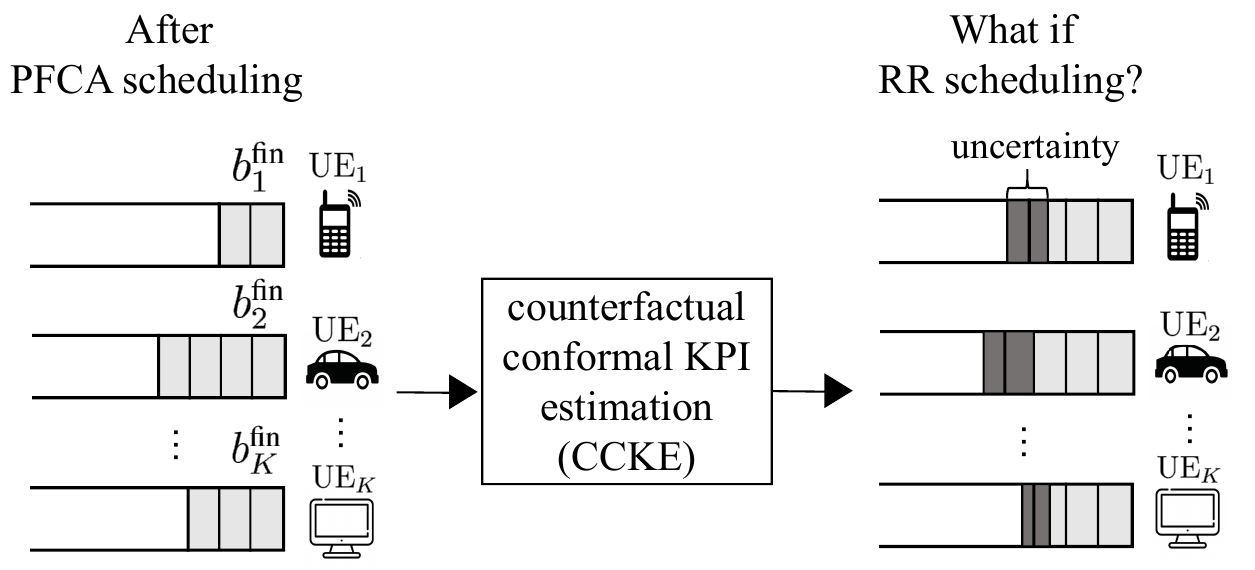}
    \label{fig:model_b}
    }
    \caption{(a) Uplink resource allocation problem in which, based on the initial backlogs and channel quality indicators (CQIs), the non-real-time controller chooses a scheduling app. (b) Given initial backlogs $b^{\mathrm{in}} = [b^{\mathrm{in}}_1,\dots,b^{\mathrm{in}}_K]$ and CQIs $c = [c_1,\dots,c_K]$ for the $K$ users, assume that the controller has selected the proportional fair channel aware (PFCA)  scheduling app, which has produced a final backlog $b^{\mathrm{fin}} = [b^{\mathrm{fin}}_1,\dots,b^{\mathrm{fin}}_K]$. What would the backlog have been if a round-robin (RR) scheduling had been selected instead?}
    \label{fig:system_model}
    \vspace{-0.5cm}
\end{figure*}

Data-driven what-if analyses are notably different from conventional prediction tasks \cite{ Wasserman book,Simeone_machine}. To design a predictor, one is typically given examples of input-output pairs that follow the same distribution of the input-output data to be predicted. In contrast, in counterfactual analysis, such data are not available. In fact, the KPI data available at the controller were recorded under contexts for which the target app was preferred by the controller to the current app. However, the testing context is one under which the target app was \emph{not} chosen by the controller, requiring a counterfactual estimate of the KPIs. This implies that there is a \emph{covariate shift} between logged data and testing data.

For the example in Fig. \ref{fig:system_model}, the data available regarding the performance of RR scheduling pertain to contexts (initial backlogs and channel quality indicators) under which RR scheduling is selected to be deployed in the RAN. However, in a counterfactual analysis, one wishes to predict the KPIs, e.g., the final backlogs, of RR under a context for which PFCA scheduling was selected. 

This problem is addressed in the literature on \emph{causal estimation}\cite{Wasserman book,causal book1,causal book2}. Causal estimation can target the \emph{average treatment effect} (ATE), the \emph{conditional average treatment effect} (CATE), or the \emph{individual treatment effect} (ITE) \cite{ITE-CP1}. The ATE measures the population-level average effect; the CATE estimates average effects for subgroups sharing similar characteristics; while the ITE estimates effects at the individual level, thus providing more granular information. In this paper, we address the ITE analysis of KPIs in wireless systems through \emph{conformal-prediction-based counterfactual inference}.

The proposed approach provides reliability guarantees by returning error bars for the KPIs to be estimated that provably contain the true KPIs with a user-defined probability. As illustrated in Fig. \ref{fig:system_model}, for the running scheduling example, the proposed method returns reliable intervals of possible values for the final backlogs.  The method builds on \emph{weighted conformal prediction} (WCP) \cite{weight-CP}, which accounts for the described {covariate shift} between logged data and testing data via a suitable weighting mechanism.

\subsection{Related Work}
\noindent \underline{\emph{Conformal prediction}}: Conformal prediction (CP) is a statistical framework \cite{CP, CP_Vladimir} that aims at calibrating the output of black-box models via  data-driven post-processing. Thanks to CP, the output of an existing predictor can be complemented with ``error bars'' that provably contain the true answer with a user-defined probability\cite{CP}. CP has been recently applied to wireless systems in \cite{CP in wireless1, CP in wireless2, CP in wireless3, NC score} to address prediction and control problems.

\noindent \underline{\emph{Reliable counterfactual analysis}}: Counterfactual analysis has been studied for years in causal analysis \cite{causal book1, causal book2}, and it generally refers to the task of answering the what-if question: What would the outcome have been had the unit been assigned to a different treatment\cite{ITE-CP1,ITE-CP2}? This is a classical task that can be viewed as an inference problem with missing data, as we can only observe one potential outcome (the actual outcome), while the others (counterfactuals) are missing\cite{ITE-CP1}. Inference with missing data has been studied for decades via statistical methods such as inverse probability weighting (IPW) \cite{causal book1} and imputation\cite{causal book1}. Both IPW and imputation-based methods aim at population-level counterfactual analysis, i.e., at the ATE or the CATE, rather than individual inferences, i.e., at the ITE. These methods are also sensitive to model misspecification, which is potentially problematic for complex systems, such as wireless networks. We refer to \cite{missing data} for further related discussions.

Recently, conformal prediction has been applied to counterfactual analysis to obtain statistical guarantees on the predicted outcomes \cite{ITE-CP1}, to perform sensitivity analysis \cite{ITE-CP2} and to obtain counterfactual estimates on the performance of off-policy reinforcement learning policies \cite{ITE-CP3}.  Conformal counterfactual analysis is an active research area in statistics with recent extensions made on \emph{multi-source} \cite{multi_source} or \emph{continuous treatment}\cite{cont_treat} scenarios. 

\noindent \underline{\emph{Counterfactual analysis in wireless systems}}: To the best of our knowledge, counterfactual analysis has been only recently mentioned as a direction for research in wireless engineering in relation to \emph{digital twinning} \cite{what-if, DT-counterfactual-wireless1}. Digital twins are simulators of real-world environments\cite{DT} that can potentially provide estimates to what-if questions such as the one posed in Fig. \ref{fig:system_model}. For example, DT platforms were leveraged in \cite{what-if,DT-AMS, DT-counterfactual-wireless1} to provide a sandbox in which different network configurations can be tested to predict their impact on KPIs. In particular, the authors of \cite{what-if} utilized a generative model for synthetic data generation at a DT to carry out a what-if analysis of KPIs including throughput, latency, packet loss, and coverage. However, the methods studied in prior work do not provide any statistical guarantee on the quality of the counterfactual estimate, which is the goal of this work.

Table \ref{table1} provides a comparison between the proposed approach and state-of-the-art counterfactual analysis methods.
\begin{table*}[htpb]
    \renewcommand\arraystretch{1.2}
	\small
	\caption{Comparison with the state of the art}
	\label{table1}
	\centering
	\scalebox{0.8}{
	\begin{tabular}{>{\color{black}}c >{\color{black}}c>{\color{black}}c>{\color{black}}c>{\color{black}}c}
		\toprule
		\textbf{Reference} & \makecell[c]{Individual-level\\ counterfactual analysis} & \makecell[c]{Finite-sample \\guarantees} & \makecell[c]{Robustness to\\ misspecification} & \makecell[c]{Application to \\wireless systems}\\
		\hline\hline
        \rowcolor{gray!30}  IPW [26] & \textcolor{black}{\ding{55}} & \textcolor{black}{\ding{55}} & \textcolor{black}{\ding{51}} & \textcolor{black}{\ding{55}}\\
        \hline
        Imputation-based methods [26] & \textcolor{black}{\ding{55}} & \textcolor{black}{\ding{55}} & \textcolor{black}{\ding{55}} & \textcolor{black}{\ding{55}}\\
        \hline
        \rowcolor{gray!30} CP-based methods [33] & \textcolor{black}{\ding{51}} & \textcolor{black}{\ding{51}} & \textcolor{black}{\ding{51}} & \textcolor{black}{\ding{55}}\\
        \hline
        \makecell[c]{Counterfactual analysis \\in wireless systems [8],[29]} & \textcolor{black}{\ding{55}}  & \textcolor{black}{\ding{55}}  & \textcolor{black}{\ding{55}} & \textcolor{black}{\ding{51}}\\
        \hline
        \rowcolor{gray!30} CCKE & \textcolor{black}{\ding{51}}  & \textcolor{black}{\ding{51}}  & \textcolor{black}{\ding{51}} & \textcolor{black}{\ding{51}}\\
		\bottomrule
	\end{tabular}} \label{table:comp_reference}
    \vspace{-0.5cm}
\end{table*}

\subsection{Main Contributions}
This paper proposes \emph{counterfactual conformal KPI estimation} (CCKE), a novel mechanism for reliable what-if KPI analysis in wireless systems based on WCP. The proposed CCKE approach provides a low-complexity post-processing solution that offers finite-sample statistical guarantees. As illustrated in Fig. \ref{fig:O-RAN}, in the system under study,  the controller collects data from the network about KPIs attained by apps selected for deployment on the basis of context information. Based on these historical data and the observation of a new set of KPIs in a test context, the controller wishes to estimate the KPIs that would have been obtained if a different app had been used.  Our main contributions are as follows.
\begin{itemize}
    \item We introduce CCKE, a statistical framework that supports the what-if analysis of KPIs in wireless systems. CCKE provides reliable error bars for the counterfactual KPI estimates, which are guaranteed to cover the true KPIs with a user-defined probability.

    \item Experimental results for the problem of uplink scheduling at the medium access control layer\cite{nokia} and for data transmission in MIMO systems at the physical layer\cite{MIMO channel model} demonstrate the effectiveness and reliability of the KPI estimates produced by CCKE.
\end{itemize}

\subsection{Organization}
The remainder of the paper is organized as follows. Section \ref{problem formulation} introduces the problem of counterfactual analysis for wireless systems. The proposed counterfactual conformal KPI estimation is presented in Section \ref{CCA}, while relevant benchmarks are described in Section \ref{baseline}. Section \ref{example2} provides an application of counterfactual analysis to medium access control-layer scheduling, and Section \ref{example1} to physical-layer transmission. Finally, Section \ref{conclusion} concludes the paper.

\section{Setting and Problem Formulation}\label{problem formulation}
In this section, we define the problem of counterfactual, or what-if, analysis for wireless systems to be studied in this paper. We summarize the main notations used throughout this paper in Table \ref{symbol_table}.

\begin{table*}[t]
    \renewcommand\arraystretch{1.2}
	\small
	\caption{Main Notations Used in This Paper}
	\label{symbol_table}
	\centering
	\scalebox{0.8}{
    
	\begin{tabular}{c|c}
		\toprule
		\textbf{Symbol} & \textbf{Definition}\\
		\hline\hline
        \makecell{$\alpha$} & miscoverage level\\
		\hline
		\makecell{$x$} & context variable\\
		\hline
		\makecell{$a$} & app identifier\\
		\hline
        \makecell{$y$} & KPIs\\
        \hline
        \makecell{$y^k$} & $k$-th element of KPIs $y$\\
        \hline
        \makecell{$y_a$} & observed KPI under context $x$ when app $a$ is selected \\
        \hline
        \makecell{$p(x)$} & probability of context $x$\\
        \hline
        \makecell{$p(a|x)$} & app selection function\\
        \hline
        \makecell{$\Gamma_{a'}(x|a)$} & prediction set for KPIs $y_{a'}$ under context $x$ and selected app $a$\\
        \hline
        \makecell{$N^{\text{tr}}$} & number of training data points\\
        \hline
        \makecell{$N^{\text{cal}}$} & number of calibration data points\\
        \hline
        \makecell{$\mathcal{D}^{\text{tr}}$} & training set\\
        \hline
        \makecell{$\mathcal{D}^{\text{cal}}$} & calibration set\\
        \hline
        \makecell{$\phi$} & parameter vector of the prediction model\\
        \hline
        \makecell{$q_\tau^k(x)$} & $\tau$-quantile of the KPI $y^k$ given context $x$\\
        \hline
        \makecell{$w_{a^{'}\rightarrow a}(x)$} & weight in WCP\\   
\bottomrule
	\end{tabular}}
	\vspace{-0.1cm}
\end{table*}

\subsection{Setting}
As illustrated in Fig. \ref{fig:O-RAN}, we study a wireless system in which the RAN is managed by a non-real-time controller. The controller logs information about the inputs used by the controller to select an app, which may be composite. It also records the identity of the selected app, and of the KPIs obtained as a result of this selection. For example, in the O-RAN architecture, inputs and KPIs can be obtained from the standardized interfaces between the RAN and the controllers \cite{intent2, O-RAN1}.

Accordingly, the controller logs data in the form 
\begin{equation}
(\underbrace{x}_{\text{context}},\underbrace{a}_{\text{apps}},\underbrace{y}_{\text{KPIs}}),
\end{equation}
where: 
\begin{itemize}
    \item The \emph{context} $x$ is a descriptor of the network operational conditions that may encompass information about the \emph{traffic}  -- e.g., number of users, the types of services being delivered, the cell average loads, fronthaul loads, and mobility levels\cite{mobile} --, \emph{connectivity} -- e.g., topology, average signal-to-noise ratio (SNR) conditions--, and \emph{multi-modal} sensory data from cameras, GPS, or radar \cite{content}. The context $x$ may also include the \emph{intents} of the network operator, possibly expressed in natural language \cite{intent1,intent2,intent3}.
    
    \item The \emph{app identifier} $a$ specifies the app being run by the network, such as scheduling algorithms, predictors, and resource allocation routines\cite{app_summary1,app_summary2,app1}. For example, in the O-RAN architecture, the operation of the RAN is managed at the time scale between $10$ ms to $1$ s by a near-real-time radio intelligence controller by plug-and-play modules known as \emph{xApps} \cite{app1}; while lower-latency applications may be run at a real-time controller on the DU and the CU by modules, known as \emph{dApps}, proposed in \cite{app2}.

    \item The \emph{KPIs} $y$ include deployment-specific performance metrics such as the latency for uRLLC traffic, the bit rate for eMBB users, and the throughput for mMTC slices\cite{KPI1,KPI2,KPI3}.
    
\end{itemize}

The goal of this work is to develop tools for the controller to answer the following \emph{what-if} question: \emph{Given that app $a$ has obtained KPIs $y$ for context $x$, what would the KPIs have been for the same context $x$ had some other app $a'\neq a$ been selected by the non-real-time controller?} Being able to address this counterfactual question would enable the controller to run diagnostic tests on its past performance, while planning for possible future policy updates.

\subsection{Problem Definition}\label{problem_definition_section}
To formalize the what-if questions of interest, assume that the controller has access to a fixed catalog of possible apps $\mathcal{A}$. Under the current context $x\in\mathcal{X}$, where $\mathcal{X}$ is the set of possible contexts, the controller selects an app $a\in\mathcal{A}$ to be run. For generality, we allow this selection to be randomized, and we accordingly describe it via the conditional probability $p(a|x)$ of choosing app $a\in\mathcal{A}$ given context $x$.

Running app $a$ under the operating condition $x$ yields KPIs $y_a\in\mathcal{Y}$, where $\mathcal{Y}$ is the set of possible values for the KPIs of interest. Consider the set of KPIs $\{y_{a'}\}_{a'\in \mathcal{A}}$, where $y_{a'}$ is the KPIs that would be attained if app $a'\in\mathcal{A}$ had been chosen. Following the standard \emph{potential outcomes} framework\cite{ITE-CP1}, the variables under study are jointly distributed according to the distribution
\begin{equation}\label{joint_distribution}
    p(x,a,\{y_{a'}\}_{a'\in\mathcal{A}}) = p(x)p(\{y_{a'}\}_{a'\in\mathcal{A}}|x)p(a|x),
\end{equation}
where $p(x)$ is the probability of context $x$; $p(\{y_{a'}\}_{a'\in\mathcal{A}}|x)$ is the conditional distribution of the potential outcomes given $x$; and $p(a|x)$ describes the mentioned controller's app selection function. The equality (\ref{joint_distribution}) uses the fact that, by definition of context $x$, the selection of the app $a$ depends exclusively on $x$. This also implies that the system under study satisfies the strong ignorability assumption\cite{ITE-CP1, ITE-CP2}.

In order to be able to estimate the potential outcome $y_{a'}$ for some app $a'$ different from the selected app $a$, we make the standard stable unit treatment value assumption (SUTVA). Accordingly, the observed KPI $y$ under context $x$ when app $a$ is selected given by the corresponding potential KPI $y_a$, i.e.,
\begin{equation}\label{consist}
    y = y_a. 
\end{equation}
This assumption is valid as long as there is no noise, adversarial or random, on the reading and reporting our KPIs. KPI observation noise can be addressed as discussed in Section \ref{WCP}.

Following reference \cite{ITE-CP1}, the objective of counterfactual analysis can be formalized as follows: Given that, under context $x\sim p(x)$, the selected app $a$ has produced KPIs $y_a$, construct a prediction set $\Gamma_{a'}(x|a)\subseteq \mathcal{Y}$ for the KPIs $y_{a'}$ that would have been obtained under the same context $x$ had app $a'\neq a$ been used instead. Note that the subscript $a'$ in the set $\Gamma_{a'}(x|a)$ indicates that the prediction set is designed for the app $a'$, while the “conditioning” on $a$ clarifies that controller has in fact run app $a$. 

For the given fixed pairs of apps $a$ and $a'$, the prediction interval $\Gamma_{a'}(x|a)$ must satisfy the coverage requirement
\begin{equation}\label{p1}
    {\Pr}\left[y_{a'}\in\Gamma_{a'}(x|a)\right]\geq 1-\alpha.
\end{equation}
By (\ref{p1}), the prediction interval $\Gamma_{a'}(x|a)$ contains the true KPIs $y_{a'}$ with probability at least $1-\alpha$, where $\alpha\in[0,1]$ is a pre-determined target miscoverage level set by the network. The probability in (\ref{p1}) is evaluated with respect to the joint distribution of all the data used to obtain the set $\Gamma_{a'}(x|a)$ and of the test data $(x,a,y)$, as detailed in Section \ref{sec:calibration}.

\subsection{Examples}\label{example_section}
To showcase the usefulness of counterfactual analysis in wireless systems, we present next two examples of applications, which will be further studied in Section \ref{example2} and Section \ref{example1} via numerical results.

\subsubsection{Scheduling at the medium access control layer} Consider the wireless multi-access channel in Fig. \ref{fig:system_model} in which multiple user equipments (UEs) share the same wireless resources. In this setting, the context $x$ may encode the backlog of the queues at each UE at the beginning of a scheduling frame, as well as the UEs' channel quality indicators (CQIs). The app $a$ is a scheduling algorithm, selected based on context $x$ from a set of scheduling algorithms, such as RR or PFCA\cite{mobile}. The KPIs $y$ may encompass the backlogs remaining at the UEs at the end of the frame. 

In this scenario, counterfactual analysis can be used to answer the question: Given that under the initial backlog $x$ the system has chosen to deploy, say, the PFCA scheduling algorithm, what would the final backlogs have been if the RR scheduling algorithm had been used instead?

\subsubsection{Multi-antenna transmission at the physical layer}
 We now consider the multi-antenna communication link in Fig. \ref{fig:SER}. In this scenario, the controller, which resides at the transmitter, determines the transmission mode and modulation scheme. This selection is done on the basis of a context $x$ encompassing information about the propagation channel, such as the estimated SNR and observations related to the richness of the multipath environment. The physical-layer transmission app $a$ implements either a multiplexing or diversity-based space-time transmission method\cite{MIMO channel model} with a modulation scheme selected from a given set of constellations. Assuming a standard automatic retransmission protocol (ARQ), the KPIs $y$ may measure the latency and/or the throughput. 
 
In this case, counterfactual analysis allows to answer the what-if question: Given that under the current channel conditions $x$ the system has implemented, say, a multiplexing-based scheme with BPSK, what would the latency have been if a diversity-based method, such as Alamouti coding with QPSK, had been used instead?

\begin{figure*}[htbp]
    \centering
    \includegraphics[width=10cm]{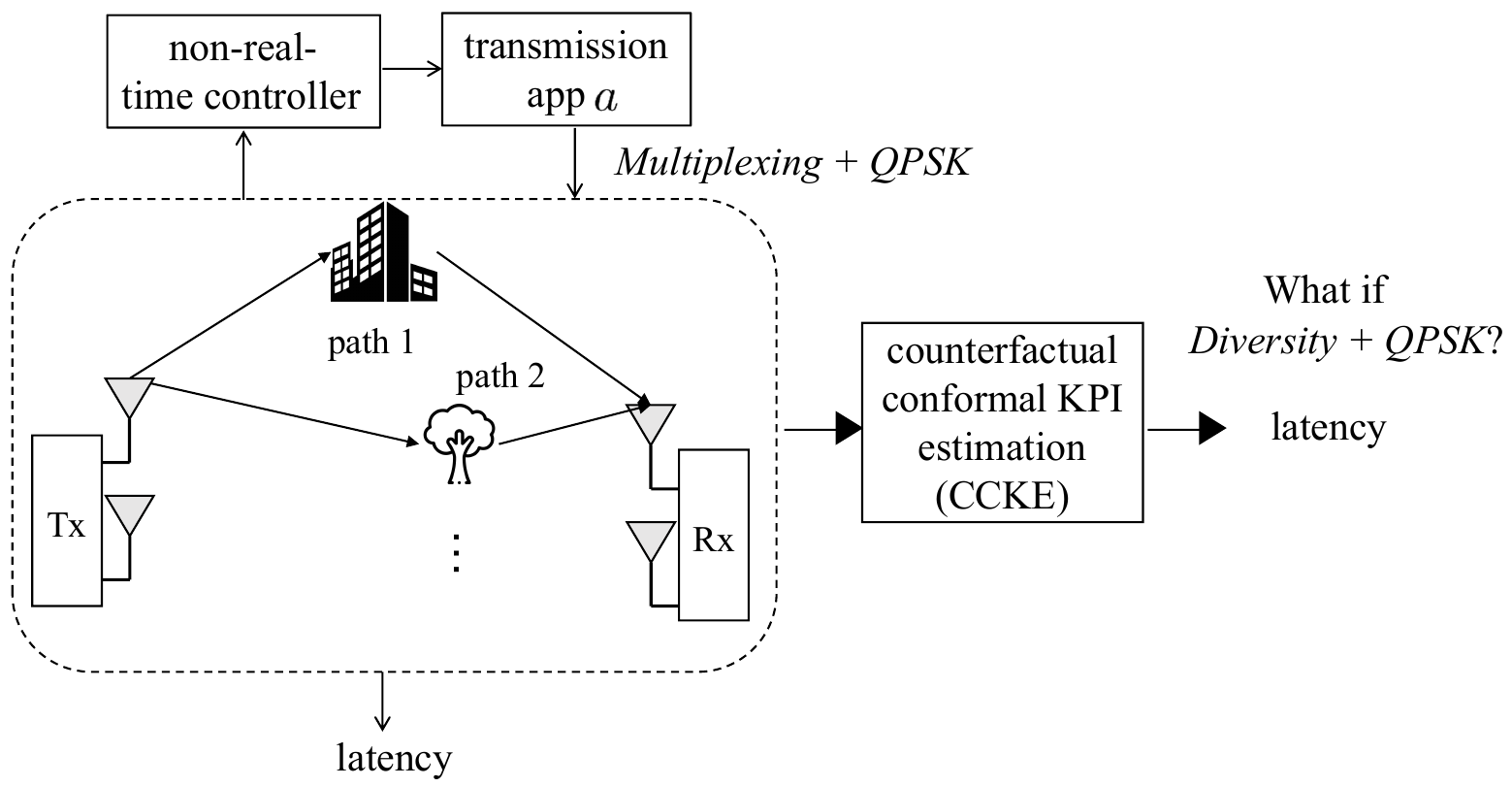}
    \vspace{-0.2cm}
    \caption{ In a multi-antenna communication link, based on the context $x$ encompassing the initial average SNR and information about the propagation environment, the non-real-time controller chooses a transmission app $a$. The transmission app $a$ is determined by a multiplexing-based or diversity-based space-time method, along with a constellation. Assuming that the controller has selected a multiplexing-based scheme with a QPSK constellation, what would the latency have been if the transmitter had used a diversity-based scheme with QPSK?}
    \label{fig:SER}
    \vspace{-0.5cm}
\end{figure*}

\subsection{Calibrating the Prediction Set}
\label{sec:calibration}
As illustrated in Fig. \ref{fig:O-RAN}, to produce a prediction set $\Gamma_{a'}(x|a)$ that satisfies the coverage condition (\ref{p1}), we assume that the controller has access to historical data $\mathcal{D} = \{(x_n,a_n,y_n)\}_{n=1}^N$, where each data entity $(x_n,a_n,y_n)$ represents an observation of 
context $x_n\sim p(x)$, app $a_n\sim p(a|x_n)$, and KPI outcome $y_n = y_{a_n}\sim p(y_{a_n}|x_n)$, for $n=1,\dots,N$. The data points in the set $\mathcal{D}$ are assumed to be independent and identically distributed (i.i.d.), i.e.,
\begin{equation}
    (x_n,a_n,y_n)\mathop{\sim}\limits^{\text{i.i.d.}} p(x)p(y_a|x)p(a|x), \text{for}\  n=1,\dots,N.
\end{equation}

At testing time, the controller observes an independent tuple
\begin{equation}\label{test_data}
    (x,a,y)\sim p(x)p(y_a|x)p(a|x),
\end{equation}
corresponding to a new context context $x$, selected app $a$, and corresponding KPI $y$. The goal of the controller is to estimate the KPI outcome $y_{a'}$ that would have been obtained by another app $a'\neq a$ by providing a prediction set $\Gamma_{a'}(x|a)$. The prediction set $\Gamma_{a'}(x|a)$ must satisfy the condition (\ref{p1}) on average over the data from the data set $\mathcal{D}$ used to construct the set $\Gamma_{a'}(x|a)$ and over the test data.

\section{Conformal Counterfactual Analysis}\label{CCA}
In this section, we introduce a scheme that leverages \emph{weighted conformal prediction} (WCP) \cite{weight-CP} to obtain a prediction set $\Gamma_{a'}(x|a)$ that meets the coverage requirement (\ref{p1}). The overall algorithm, referred to as \emph{counterfactual conformal KPI estimation} (CCKE), is described in Algorithm \ref{A1}, and an illustration is provided in Fig.~\ref{fig:cartoon}, with details explained in the rest of this section. 

\begin{figure*}[htbp]
    \centering
    \includegraphics[width=14cm]{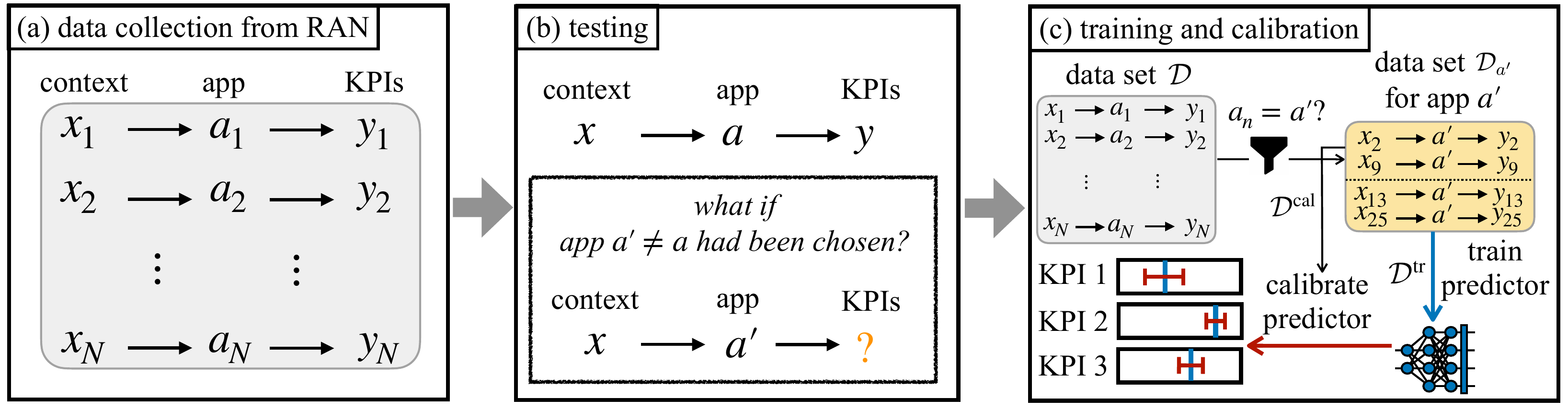}
    \vspace{-0.2cm}
    \caption{Overview of CCKE: (a) The RAN collects a data set $\mathcal{D}$ consisting of context $x_n$, selected app $a_n$, and corresponding observed KPIs $y_n$ for $n=1,...,N$. (b) Given the new context $x$, the controller selects app $a$ and the RAN observes the corresponding KPIs $y$. CCKE aims at reliably estimating the observed KPIs if app $a'\neq a$ had been chosen instead of $a$ for the current context $x$. (c) To do so, CCKE first extracts data set $\mathcal{D}_{a'}$ from the entire data set $\mathcal{D}$ by choosing only the data points corresponding to app $a'$. Then, the first split of the data set is used to train the predictor (\ref{CQR}), while the remaining data are used for calibration via (\ref{NC score}). The resulting calibrated prediction interval (\ref{w_interval_calculate}) provably contains the true KPIs $y_{a'}$ for the target app $a'$ of interest (Lemma~\ref{lemma}).   }
    \label{fig:cartoon}
    \vspace{-0.2cm}
\end{figure*}

\subsection{Addressing Multiple KPIs}
Reflecting settings of practical interest, we assume henceforth that the KPI $y\in\mathcal{Y}$ consists of $k$ scalar KPIs $y=[y^1,\dots,y^K]$, so that the set $\mathcal{Y}$ coincides with the $k$-dimensional real space $\mathbb{R}^K$. Accordingly, our goal is to produce $K$ prediction intervals $\{\Gamma_{a'}^k(x|a)\}_{k\in[K]}$, collectively denoted as $\Gamma_{a'}(x|a)$, such that the following simultaneous coverage condition is satisfied (cf. (\ref{p1}))
\begin{equation}\label{p1_multi}
    \Pr[y^k_{a'}\in\Gamma_{a'}^k(x|a)\ \text{for all}\  k=1,\dots,K]\geq 1-\alpha.
\end{equation}
The coverage guarantee  (\ref{p1_multi}) ensures that the true KPI vector $y_{a'} = [y^1_{a'},\dots, y^K_{a'}]$ that would have been obtained under app $a'$ for context $x$ is in set $\Gamma_{a'}(x|a) = \{\Gamma_{a'}^k(x|a)\}_{k\in[K]}$ with a probability no smaller than $1-\alpha$.

\subsection{Data Selection}
To produce the KPIs prediction set $\Gamma_{a'}(x|a)$ for the target app $a'$,  CCKE first selects the subset of data points from the data set $\mathcal{D}$ that were collected for the target app $a'$, i.e., data points of the form $(x,a',y)$. We denote the resulting data set as ${\mathcal{D}}_{a'}=\{(x_n,y_n)\}_{n=1}^{{N}_{a'}}$, where we have dropped the app identifier $a'$ from the selected 
 tuple $(x_n,a_n=a',y_n)$ to avoid clutter. Note that ${N}_{a'}$ is the number of instances that the app $a'$ was selected, i.e., ${N}_{a'} = \sum_{n=1}^N \mathbbm{1}(a_n=a')$.

The resulting data set ${\mathcal{D}}_{a'}$ is partitioned into a training set $\mathcal{D}^{\text{tr}}$ with $N^{\text{tr}}$ data points and a calibration set $\mathcal{D}^{\text{cal}}={\mathcal{D}}_{a'}\setminus\mathcal{D}^{\text{tr}}$ with $N^{\text{cal}} = N_{a'}-N^{\text{tr}}$ data points. Note that we drop the dependence of data sets $\mathcal{D}^{\text{tr}}$ and $\mathcal{D}^{\text{cal}}$ on the app identifier $a'$ to simplify the notation. As described in the next subsection, the training set $\mathcal{D}^{\text{tr}}$ is used to train a prediction model for the KPI $y$ given the context $x$. By construction, the trained prediction model targets data distributed as $(x,y)\sim p(x|a')p(y_{a'}|x)$, and thus it aims at predicting the KPIs attained by the target app $a'$. The calibration set $\mathcal{D}^{\text{cal}}$ is used to calibrate the output of the trained predictor in order to obtain the set predictor $\Gamma_{a'}(x|a)$, as detailed in Section \ref{WCP}. 

\subsection{Prediction Model Training}

CCKE adopts a quantile regression model to estimate the $\alpha/2$- and $(1-\alpha/2)$-quantiles of KPIs $\{y^k\}^K_{k=1}$ for any context $x$ based on the training data set $\mathcal{D}^{\text{tr}}$, which reports the KPIs of interest, $\{y_n^k\}_{k=1}^K$, under different contexts $x_n$ with $n=1,\dots,N^{\text{tr}} $\cite{CQR}.  

For $\tau\in[0,1]$, the $\tau$-quantile of the KPI $y^k$ given context $x$ is defined as
\begin{align}
    \label{eq:kpi_quantile}
    q^k_{\tau}{(x)}=\mathrm{inf}\{q^k\in\mathbb{R}: \Pr[y^k\leq q^k {|x}]\geq \tau\},
\end{align}
with a probability taken with respect to the conditional distribution of random variable $y^k$ given $x$. The quantile (\ref{eq:kpi_quantile}) is equivalent to the minimizer of the \emph{pinball} loss \cite{pinball loss}, i.e.,  
\begin{equation}\label{eq:expectation of KPI}
    q_\tau^k(x) = \arg\min_{\hat{q}^k}\mathbb{E}\big[\ell_\tau(y^k, \hat{q}^k)|x\big]
\end{equation}
where
\begin{equation}
    \ell_\tau(y^k, \hat{q}^k) = \mathop{\rm{max}} \left\{\tau (y^k-q^k), -(1-\tau)(y^k-q^k)\right\},
\end{equation}
and the expectation in (\ref{eq:expectation of KPI}) is taken with respect to the conditional distribution of $y^k$ given $x$. By (\ref{eq:expectation of KPI}), we can equivalently write
\begin{align}\label{eq:quantile problem}
    [q_\tau^1(x),...,q_\tau^K(x)] = \arg\min_{\hat{q}^1,...,\hat{q}^K}\sum_{k=1}^K \mathbb{E}\big[\ell_\tau(y^k, \hat{q}^k)|x\big].
\end{align}

Quantile regression \cite{pinball loss} can be leveraged to estimate quantiles by addressing problem (\ref{eq:quantile problem}) over the parameter vector $\phi$ of parametric functions $\{\hat{q}_{\tau}^k(\cdot|\phi)$ for $k=1,\dots,K\}$. Specifically, estimating the averages in (\ref{eq:quantile problem}) using data set $\mathcal{D}^{\text{tr}}$, one obtains the training problem
\begin{align}
    \phi_\tau = \arg\min_\phi \bigg\{ \sum_{n=1}^{N^\text{tr}} \sum_{k=1}^K \ell_\tau(y_n^k,  \hat{q}_\tau^k(x_n|\phi)  )  \bigg\}.
\end{align}
For the purpose of ensuring the reliability condition (\ref{p1}), we are specifically interested in estimating the $\alpha/2$- and $(1-\alpha/2)$-quantiles, yielding the final training problem

\begin{equation}\label{CQR}
    \phi^* = \mathop{\rm{argmin}}\limits_{\phi} \left\{\sum_{n=1}^{N^{\text{tr}}}\sum_{k=1}^{K}\bigg(\ell_{\alpha/2}(y_n^k,\hat{q}_{\alpha/2}^k(x_n|\phi))\right.
        \left.+\ell_{1-\alpha/2}(y_n^k,\hat{q}_{1-\alpha/2}^k(x_n|\phi))\bigg)\right\}.
\end{equation}

Given the trained parameter vector $\phi^*$, the $\alpha/2$-quantile $q_{\alpha/2}^k(x)$ is estimated as $\hat{q}_{\alpha/2}^k(x|\phi^*)$, while the $(1-\alpha/2)$-quantile $q_{1-\alpha/2}^k(x)$ is estimated as $\hat{q}_{1-\alpha/2}^k(x|\phi^*)$, for any context $x \in \mathcal{X}$.

With these quantile estimates, given a new context $x$, the predicted interval for each KPI $y^k$, with $k=1,\dots, K$, is obtained as
\begin{equation}\label{naive prediction interval}
   \tilde{\Gamma}_{a'}^k(x) = [\hat{q}_{\alpha/2}^k(x|\phi^*), \hat{q}_{1-\alpha/2}^k(x|\phi^*)].
\end{equation}
However, the intervals (\ref{naive prediction interval}) generally fail to satisfy the coverage condition (\ref{p1}), since (\emph{\romannumeral1}) the model $\hat{q}_\tau(x|\phi^*)$ may not be expressive enough to capture the true function mapping context into quantiles; and (\emph{\romannumeral2}) the amount of training data $N^{\text{tr}}$ is limited, not representing the full joint distribution of contexts and KPIs. The lack of coverage guarantee for the intervals (\ref{naive prediction interval}) is addressed by the proposed CCKE by utilizing the calibration data $\mathcal{D}^\text{cal}$, as explained next.

\subsection{Calibration via Weighted Conformal Prediction}\label{WCP}
In this subsection, we explain how to apply WCP\cite{weight-CP} to calibrate each prediction interval in (\ref{naive prediction interval}), so that it satisfies the coverage condition (\ref{p1}). 

A key challenge in calibrating the prediction interval is that the context variables $x$ in the calibration data set $\mathcal{D}^{\text{cal}}$ are drawn from a different distribution as compared to the context $x$ in the test data $(x,a,y)$. In fact, the contexts $x$ in the calibration data set $\mathcal{D}^{\text{cal}}$ follow the distribution $p(x|a') = p(x)p(a'|x)/p(a')$, depending on the target app $a'$, while the context $x$ in the test tuple follows the distribution $p(x|a) = p(x)p(a|x)/p(a)$, which depends on the actual app $a$ that was run under the current context $x$.

The \emph{covariate shift} between the calibration distribution $p(x|a')$ and the test distribution $p(x|a)$ arises from the fact that the controller selects apps $a$ based on the context $x$. If the controller chose the app $a$ in a manner independent of the context $x$, one would have the equality thus $p(x|a) = p(x|a') = p(x)$. However, in practice, a controller that selects apps $a$ without considering the context $x$ is not likely to perform well. For example, with the physical-layer transmission setting in Section \ref{example_section}, the transmission scheme should be adapted to the channel condition.

Following CP\cite{CP,VB-CP}, using the preliminary prediction intervals (\ref{naive prediction interval}), WCP first computes a score for each $n$-th calibration data point $(x_n,y_n)$ as \cite{CQR, CQR_2}
\begin{equation}\label{NC score}
    s_n = \mathop{\rm{max}}\limits_{k=1,\dots,K} \mathop{\rm{max}} \left\{\hat{q}_{\alpha/2}^k(x_n|\phi^*)-y_n^k, y_n^k-\hat{q}_{1-\alpha/2}^k(x_n|\phi^*)\right\}.
\end{equation}
The score (\ref{NC score}) measures the worst-case error made when using the uncalibrated prediction intervals (\ref{naive prediction interval}) across all KPIs. In fact, the score $s_n$ is positive when at least one of the true KPIs $y_n^k$ is outside the interval (\ref{naive prediction interval}), increasing as $y_n^k$ moves further away from the interval; while it is negative otherwise, decreasing as $y_n^k$ moves further within the set. The maximum operator over the $K$ KPIs in (\ref{NC score}) is crucial to ensure the simultaneous coverage condition (\ref{p1_multi})\cite{NC score}.

WCP then evaluates a correction for the prediction set (\ref{naive prediction interval}) that satisfies (\ref{p1_multi}). To this end, it produces the sets $\{\Gamma_{a'}^k(x|a)\}_{k\in[K]}$, where
\begin{equation}\label{w_interval_calculate}
    \Gamma_{a'}^k(x|a) = \big[\hat{q}_{\alpha/2}^k(x|\phi^*)-\hat{q}_w, \hat{q}_{1-\alpha/2}^k(x|\phi^*)+\hat{q}_w\big]
\end{equation}
for $k=1,\dots,K$, given a correction term $\hat{q}_w$ optimized using the calibration data. Specifically, the correction term $\hat{q}_w$ is evaluated as the $(1-\alpha)(N^{\text{cal}}+1)/N^{\text{cal}}$-quantile -- i.e., the $\lceil(1-\alpha)(N^{\text{cal}}+1)\rceil$-th smallest value -- of a \emph{weighted empirical distribution} of the calibration scores $\{s_n\}_{n=1}^{N^{\text{cal}}}$, with weights accounting for the covariate shift. 

Mathematically, the weighted empirical distribution is obtained as 
\begin{equation}\label{weight distribution}
    \sum_{n=1}^{N^{\text{cal}}}p_n(x)\delta_{s_n}+p_{\infty}(x)\delta_{\infty},
\end{equation}
where $\delta_{s_n}$ denotes a distribution that places all probability mass at the value $s_n$, and the probabilities $\{p_1(x),\dots,p_{N^{\text{cal}}}(x),p_\infty(x)\}$ are evaluated as
\begin{align}
    p_n(x) &= \frac{w_{a' \rightarrow a}(x_n)}{\sum_{n=1}^{N^{\text{cal}}}w_{a' \rightarrow a}(x_n)+w_{a' \rightarrow a}(x)},\label{no_ca}\\
    p_\infty(x) &= \frac{w_{a' \rightarrow a}(x)}{\sum_{n=1}^{N^{\text{cal}}}w_{a' \rightarrow a}(x_n)+w_{a' \rightarrow a}(x)}\label{no_test},
\end{align}
with weights obtained by the density ratio as
\begin{equation}\label{weight}
    w_{a' \rightarrow a}(x) = \frac{p(a|x)}{p(a'|x)}.
\end{equation}
The weights (\ref{weight}) can be evaluated by using the controller's app selection function $p(a|x)$. Including the density ratio (\ref{weight}) in the empirical distribution (\ref{weight distribution}) compensates for the covariate shift between testing data and the calibration data, thus guaranteeing the coverage condition (\ref{p1}) (see, e.g., \cite{weight-CP, doubly_CP}). Note that the correction term $\hat{q}_w$ can now be explicitly written as (cf. (\ref{eq:kpi_quantile}))
\begin{equation}
    \hat{q}_w = \mathrm{inf} \left\{s\in\mathbb{R}\cup \{+\infty\}: \sum_{n=1}^{N^{\text{cal}}} p_n(x) \cdot \mathbbm{1}(s_n\leq s) + p_\infty(x) \mathbbm{1}(\infty\leq s) \geq \frac{(1-\alpha)(N^{\text{cal}}+1)}{N^{\text{cal}}}\right\}
\end{equation}
with the convention that $\mathbbm{1}(\infty\leq\infty) = 1$, where $\mathbbm{1}(\cdot)$ is the indicator function ($\mathbbm{1}(\text{true}) = 1$ and $\mathbbm{1}(\text{false}) = 0$).

Note that CCKE only requires computing an empirical quantile, an operation whose complexity is the same as for ordering, i.e., $O(N^{\text{cal}}{\rm{log}}N^{\text{cal}})$\cite[Chapter 8]{cormen2009introduction}. Therefore, CCKE is easily scaled to large networks. The overall CCKE algorithm is summarized in Algorithm \ref{A1}, and it provides the following guarantees.

\subsection{Theoretical Guarantees}
Assuming exact knowledge of the density ratio (\ref{weight}) and wireless KPI observation (\ref{consist}), CCKE satisfies the following guarantee.

\begin{lemma}[Coverage guarantee of CCKE]\label{lemma}
    For any $\alpha\geq 1/(N^{\text{cal}}+1)$, the prediction set (\ref{w_interval_calculate}) produced by CCKE satisfies the inequality (\ref{p1_multi}) with probability evaluated with respect to the joint distribution of calibration and test data for any two given app identifiers $a$ and $a'$.
\end{lemma}
\begin{proof}
    This result follows directly from {\protect{\cite[Proposition 1]{ITE-CP1}}}.
\end{proof}

Assume now that the density ration $w_{a'\rightarrow a}(x)$ in (\ref{weight}) is only known via an estimate $\hat{w}_{a'\rightarrow a}(x)$.

\begin{lemma}[Coverage guarantee of CCKE with estimated weight $\hat{w}_{a'\rightarrow a}(x)$]\label{lemma2}
    For any $\alpha\geq 1/(N^{\text{cal}}+1)$, the prediction set (\ref{w_interval_calculate}) produced by CCKE with estimated weight $\hat{w}_{a'\rightarrow a}(x)$ in lieu of the true density ratio $w_{a'\rightarrow a}(x)$ satisfies the inequality
    \begin{equation}\label{weight mismatch}
        \Pr\left[y_{a'}^k \in \Gamma_{a'}^k(x|a)\ \text{for all}\  k = 1,\dots, K\right] \geq 1 - \alpha - \frac{1}{2}\mathbb{E}_{x \sim p(x|a')}|\hat{w}_{a'\rightarrow a}(x) - w_{a'\rightarrow a}(x)|.
    \end{equation}
\end{lemma}
\begin{proof}
    This result follows from {\protect{\cite[Appendix A.1]{ITE-CP1}}}.
\end{proof}\\
The result (\ref{weight mismatch}) shows that, with an imperfect estimate $\hat{w}_{a'\rightarrow a}(x)$, the coverage is decreased from the nominal value $1-\alpha$ by an amount that depends on the average of the error $|\hat{w}_{a'\rightarrow a}(x) - w_{a'\rightarrow a}(x)|$.

Finally, performance guarantees can be also extended to the case in which the KPI observation is given by 
\begin{equation}\label{noise_equation}
    \tilde{y} = y_a+\epsilon,
\end{equation}
where $\epsilon$ is a noise variable, which is independent of all other variables.

\begin{lemma}[Coverage guarantee of CCKE with noisy KPI measurements]\label{lemma3}
    Assume that the observation noise is not always positive or negative, i.e., that we have the inequality
    \begin{equation}
        b = {\rm{min}} \{p(\epsilon \geq 0|x), p(\epsilon \leq 0|x)\}\textgreater 0.
    \end{equation}
    Then, for any $\alpha\geq 1/(N^{\text{cal}}+1)$, the prediction set (\ref{w_interval_calculate}) produced by CCKE using observations in (\ref{noise_equation}) satisfies the coverage guarantee
    \begin{equation}\label{noise guarantee}
    \Pr[y_{a'}\in\Gamma_{a'}(x|a)]\geq 1-\frac{\alpha}{b}.
    \end{equation}
\end{lemma}
\begin{proof}
    This result follows directly from {\protect{\cite[Proposition 2.3]{noise CP}}}.
\end{proof}\\
The inequality (\ref{noise guarantee}) demonstrates that, due to the observation noise, the miscoverage rate increases to $\alpha/b$. Thus, the coverage degradation is more pronounced for more skewed noise distributions. In particular, for a symmetric distribution, the miscoverage rate is, at most, doubled to $2\alpha$.

\begin{table}[htp]
	\setlength{\abovecaptionskip}{-2pt}
	\setlength{\belowcaptionskip}{-6pt}
	\begin{algorithm}[H]
		\caption{\textbf{Counterfactual Conformal KPI Estimation} (CCKE)}
		\label{A1}
		{\normalsize
			\begin{algorithmic}[1]
               \State \textbf{Input:} Miscoverage level $\alpha$; data set $\mathcal{D}=\{(x_n,a_n,y_n)\}_{n=1}^{N}$; test data $(x,a,y)$
               \State \textbf{Output:} Prediction interval $\Gamma_{a'}(x|a)$ with coverage guarantee (\ref{p1})

               \State Select all the data point of the form $(x,a',y)$ from the data set $\mathcal{D}$ to construct the data set $\mathcal{D}_{a'} =\{(x_n,y_n)\}_{n=1}^{{N}_{a'}}$
    
               \State Partition data set $\mathcal{D}_{a'}$ into a training set $\mathcal{D}^{\text{tr}}$ and a calibration set $\mathcal{D}^{\text{cal}}$
               \State Train a prediction model on training set $\mathcal{D}^{\text{tr}}$ using (\ref{CQR}), obtaining the uncalibrated prediction interval (\ref{naive prediction interval})
               \State Compute the scores $\{s_n\}_{n=1}^{N^{\text{cal}}}$ in (\ref{NC score}) for all data points in $\mathcal{D}^{\text{cal}}$
               \State Evaluate correction term $\hat{q}_w$ as the $(1-\alpha)(N^{\text{cal}}+1)/N^{\text{cal}}$-quantile of the weighted empirical distribution in (\ref{weight distribution}) using (\ref{no_ca})-(\ref{weight})
               \State Obtain the final prediction set $\Gamma_{a'}(x|a)$ in (\ref{w_interval_calculate})
               
            \end{algorithmic}}
	\end{algorithm}
\end{table}

\section{Baselines and Performance Measures}\label{baseline}
 We now introduce the relevant baselines and performance metrics to evaluate the benefits and effectiveness of the proposed CCKE method.
\subsection{Baselines}
We will consider the following baselines:
\begin{itemize}

    \item \textbf{Counterfactual KPI Estimation} (CKE): CKE adopts the prediction intervals $\{\tilde{\Gamma}_{a'}^k(x)\}_{k\in[K]}$ in (\ref{naive prediction interval}) without the proposed correction based on calibration, i.e., running from step 1 only up to step 5 in Algorithm \ref{A1}.

    \item \textbf{Na\"ive CCKE} (NCCKE): NCCKE attempts to calibrate the prediction sets, but it neglects the problem of covariate shift described in Section \ref{WCP}. Accordingly, it produces the prediction sets (\ref{w_interval_calculate}) with the correction term $\hat{q}_w$ in (\ref{w_interval_calculate}) evaluated as the $(1-\alpha)(N^{\text{cal}}+1)/N^{\text{cal}}$-quantile of the empirical distribution
    \begin{equation}\label{NCCKE distribution}
    \sum_{n=1}^{N^{\text{cal}}}\frac{1}{N^{\text{cal}}+1}\delta_{s_n}+\frac{1}{N^{\text{cal}}+1}\delta_{\infty}.
    \end{equation}
    By disregarding the weights in (\ref{weight}), NCCKE does not satisfy the reliability guarantee (\ref{p1})\cite{weight-CP}.
\end{itemize}

\subsection{Performance Metrics}

For all the schemes under study, considering fixed pairs of apps $a$ and $a'$, we report their empirical coverage and empirical inefficiency evaluated on the test data set $\mathcal{D}^{\text{te}} = \{x_n,a_n=a,y_n\}_{n=1}^{N^{\text{te}}}$ of size $N^{\text{te}}=100$. Specifically, the empirical coverage measures the fraction of times in which all KPIs were covered by the respective prediction intervals, i.e.,
\begin{equation}\label{coverage}
    \text{empirical coverage} = \frac{1}{N^{\text{te}}}\sum_{n=1}^{N^{\text{te}}}\mathbbm\prod_{k=1}^K\mathbbm{1}\left(y_n^k\in \Gamma_{a'}^k(x_n|a)\right),
\end{equation}
The normalized empirical inefficiency is defined as
\begin{small}
\begin{equation}\label{size}
     \text{normalized empirical inefficiency} = \frac{1}{N^{\text{te}}}\sum_{n=1}^{N^{\text{te}}}\frac{1}{K}\sum_{k=1}^K\frac{|\Gamma_{a'}^k(x_n|a)|}{\epsilon_n}.
\end{equation}
\end{small}
By (\ref{size}), the normalized empirical inefficiency measures the average size of the prediction sets normalized by a problem-specific constant $\epsilon$ that will be specified for the two example applications in Section \ref{example2} and Section \ref{example1}.

We average the empirical coverage and empirical inefficiency over $200$ experiments, each corresponding to independent draws of the calibration and test data set pair $\{\mathcal{D}^{\text{cal}}, \mathcal{D}^{\text{te}}\}$. This way, the metric (\ref{coverage}) provides an estimate of the coverage probability in (\ref{p1}) and of the normalized size of prediction set $\mathbb{E}[1/K\sum_{k=1}^K\Gamma_{a'}^k(x|a)/{\epsilon}]$, respectively\cite{estimate}.

The code to reproduce all the results in the following sections is available at\\ \href{https://github.com/qiushuo0913/Inference\_DT\_code}{https://github.com/qiushuo0913/Inference\_DT\_code}.

\section{Application to Medium Access Control-Layer Scheduling}\label{example2}
In this section, we elaborate further on the scheduling setting described in Section \ref{example_section}, and demonstrate the operation of CCKE through simulation results.

\subsection{Setup}
\subsubsection{Setting}
As illustrated in Fig. \ref{fig:system_model}, we consider an uplink resource allocation task in an orthogonal frequency division multiplexing (OFDM)-based multi-access system with $K$ UEs. At the beginning of a scheduling frame, each $k$-th UE has a backlog of $b^{\mathrm{in}}_k$ packets in its own queue. Furthermore, each $k$-th UE has a CQI $c_k$ that describes the quality of its connection to the base station. Accordingly, context $x$ collects the backlogs of all the $K$ UEs at the beginning of the scheduling frame, $\{b^{\mathrm{in}}_k\}_{k=1}^K$, as well as the all CQIs for $K$ UEs, $\{c_k\}_{k=1}^K$, i.e., 
\begin{equation}\label{context}
    x=\{b^{\mathrm{in}}_k,c_k\}_{k=1}^K.
\end{equation}
App $a$ can be selected as one of two standard scheduling algorithms, namely RR and PFCA\cite{mobile}. Finally, KPIs $y$ represent the vector of backlogs remaining at the $K$ UEs at the end of the scheduling frame, i.e., $y = \{b^{\mathrm{fin}}_k\}_{k=1}^K$.

The system, including the generation of initial backlogs, is implemented using the 5G emulator provided by Nokia available at \cite{nokia}.

Following Section \ref{example_section}, the selection of the app $a\in\{\text{RR},\text{PFCA}\}$ at the controller 
is modeled via the conditional distribution
\begin{equation}\label{e_x}
    p(a=\text{RR}|x) = \frac{\exp(-{\hat{b}^{\mathrm{fin}}_x}/{T})}{1+\exp(-{\hat{b}^{\mathrm{fin}}_x}/{T})}
\end{equation}
with parameter $T$, where $\hat{b}^{\mathrm{fin}}_x$ is an estimate of the largest remaining backlog if scheduling app RR was chosen under context $x$, i.e.,
\begin{equation}
    \hat{b}^{\mathrm{fin}}_x = \mathop{\rm{max}}\limits_{k=1,\dots,K} \hat{b}^{\mathrm{fin}}_{k,x},
\end{equation}
where $\hat{b}^{\mathrm{fin}}_{k,x}$ is the estimate of the final backlog for UE $k$. The rationale for focusing on the probability of selecting RR in (\ref{e_x}) is that an estimate of the residual backlog can be obtained as
\begin{equation}\label{RR_remain}
    \hat{b}^{\mathrm{fin}}_{k,x} = b^{\mathrm{in}}_k-\frac{g(c_k)}{K},
\end{equation}
given the function $g(\cdot)$ that maps CQI $c_k$ to the expected transmission payload if the entire frame was allocated to user $k$. This can be obtained following Table 7.2.3-1 from TS 36.213 Rel-11\cite{3GPP}, and the division by $K$ in (\ref{RR_remain}) accounts for the fact that, with RR, each UE is allocated a fraction $1/K$ of the spectral resources. Note that (\ref{RR_remain}) is just an estimate of the true final backlog, as the true final backlog depends on the random evolution of the system.

As illustrated in Fig. \ref{fig:e_x}, the parameter $T \textgreater 0$ in (\ref{e_x}) makes it possible to control the dependence of the selected app on the context $x$. For $T\rightarrow 0$, the app selection becomes increasingly dependent on the context $x$. In particular, if $\hat{b}^{\mathrm{fin}}_x\leq 0$, and thus RR is expected to succeed in serving all the backlog, the controller chooses $a=\text{RR}$ with probability tending to $1$ as $T \rightarrow 0$. At the other extreme, for $T\rightarrow \infty$, the app selection mechanism at the controller becomes less dependent on the context $x$. In fact, as $T\rightarrow \infty$, the probability (\ref{e_x}) tends to $0.5$, and hence both apps are selected with the same probability.

\begin{figure}[htbp]
    \centering
    \includegraphics[width=6cm]{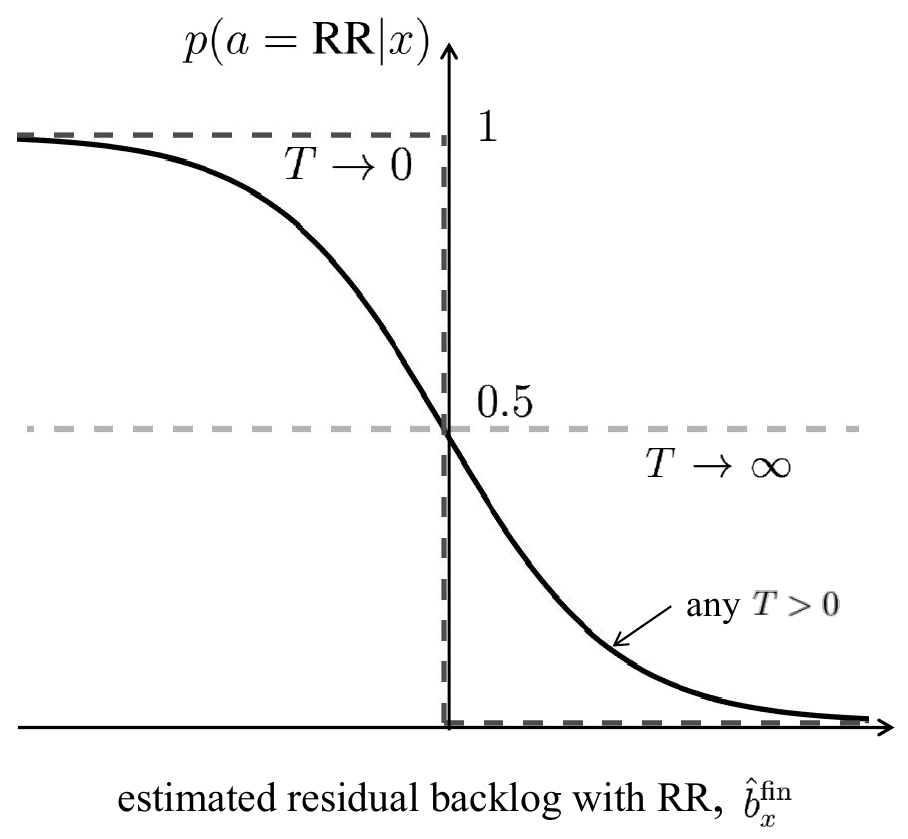}
    \vspace{-0.2cm}
    \caption{Illustration of the app selection probability (\ref{e_x}) applied by the non-real-time controller for the medium access control scheduling example.}
    \label{fig:e_x}
    \vspace{-0.2cm}
\end{figure}

\subsubsection{Prediction Model}
The quantile prediction model $\hat{q}_\tau(x|\phi)$ in (\ref{CQR}) must produce an estimate of the $K$ residual backlogs, $y = \{b_k^{\mathrm{fin}}\}_{k=1}^K$. To this end, we adopt a self-attention mechanism that ensures the important property of permutation equivariance with respect to the ordering of the $K$ UEs, i.e., $\hat{q}_\tau\left(\pi(x)|\phi\right) = \pi\left(\hat{q}_\tau(x|\phi)\right)$ for any permutation operator $\pi: (1,\dots,K) \rightarrow (1,\dots,K)$. This is detailed in the Appendix. The quantile regressor is trained using the $N^{\mathrm{tr}} = 3000$ training samples, while  $N^{\mathrm{cal}} = 50$ data pairs are used for calibration.

The parameter $\epsilon$ is set as the maximum initial backlog across the $K$ UEs, i.e., $\epsilon_n = \mathop{\rm{max}}\limits_{k} \{b_k^{\mathrm{in},n}\}$, and we refer to Appendix for the detailed parameters.

\subsection{Performance Analysis}
In this subsection, we evaluate CCKE coverage and inefficiency  in a scenario in which the actual app selected by the controller is the PFCA algorithm, i.e., $a=\mathrm{PFCA}$. The goal is thus predicting the final backlogs that would have been observed if the RR algorithm had been selected instead.

\begin{figure}[htp]
    \centering
    \subfigure[Coverage performance (\ref{coverage})]{
    \includegraphics[width=7.5cm]{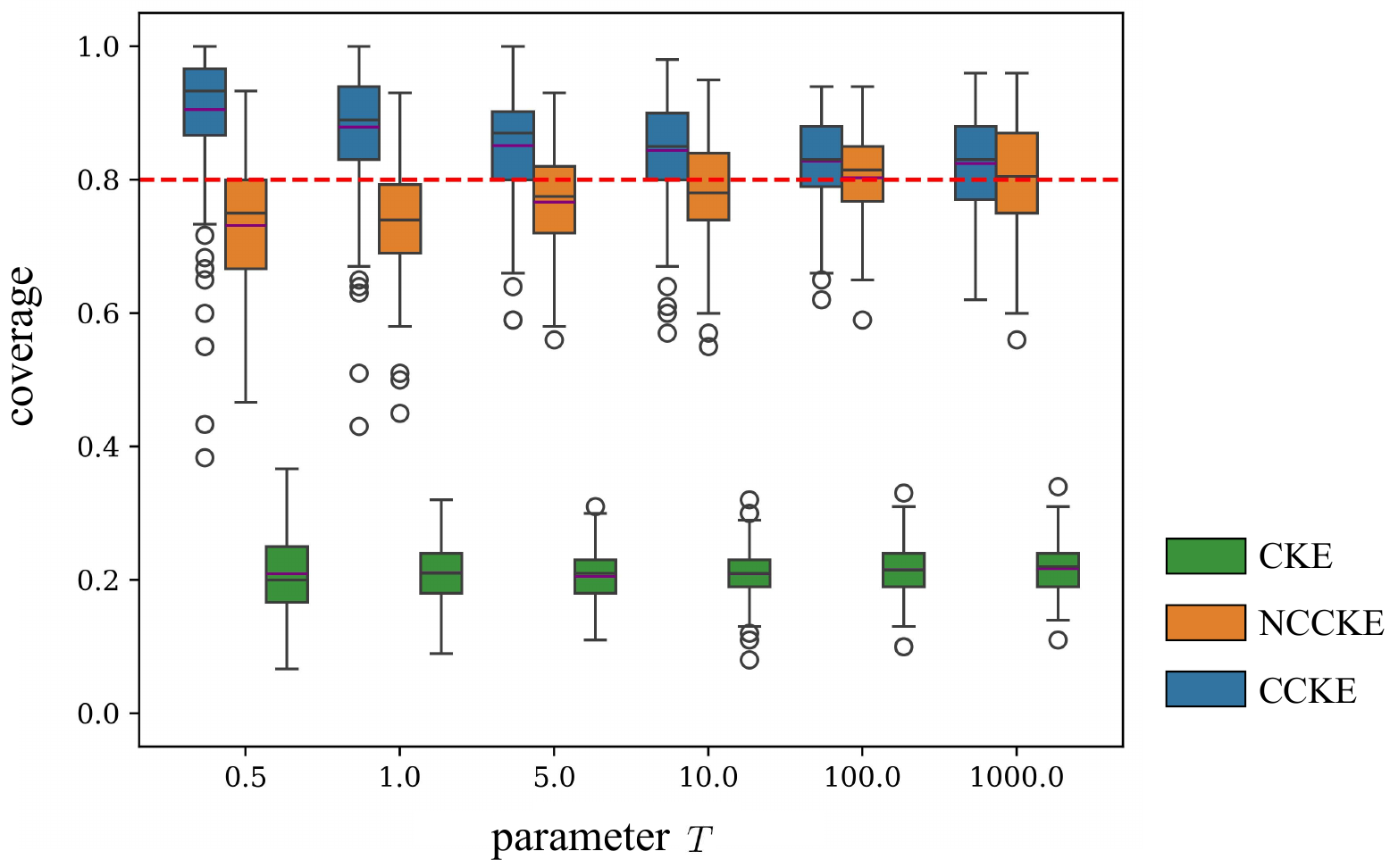}
    \vspace{-0.5cm}
    \label{fig:coverage_0}
    }
    \subfigure[Inefficiency performance (\ref{size})]
    {
    \centering
    \includegraphics[width=7.5cm]{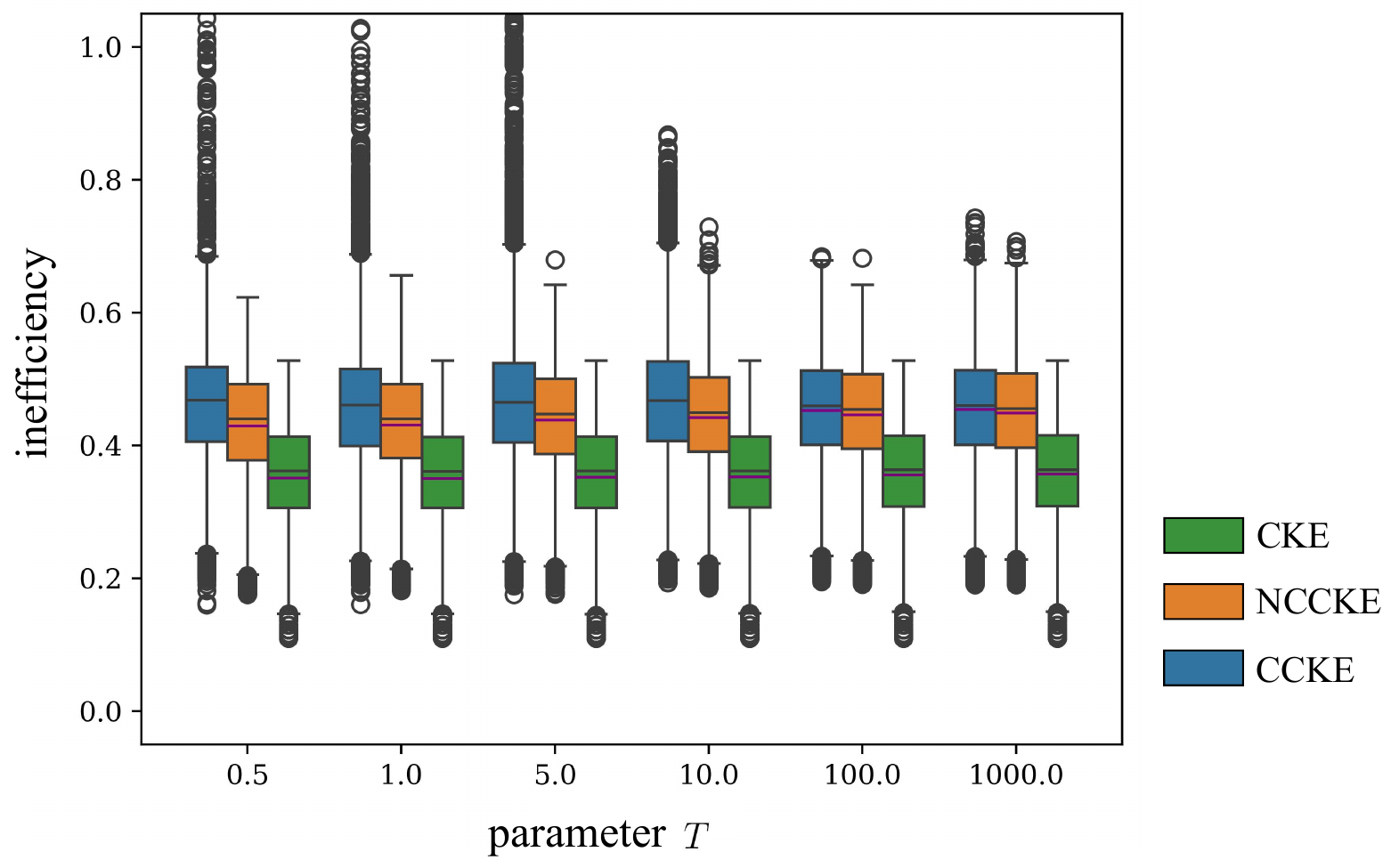}
    \label{fig:size_0}
    }
    \caption{Coverage and inefficiency performance for the counterfactual KPI analysis of the scheduling app RR given that the actual scheduling app is PFCA as a function of parameter $T$ in (\ref{e_x}). The target coverage level is $1-\alpha = 0.8$ (dashed line), and the results are averaged over $200$ experiments.}
    \label{fig: coverage_size}
    \vspace{-0.3cm}
\end{figure}

Fig. \ref{fig:coverage_0} and Fig. \ref{fig:size_0} present the coverage (\ref{coverage}) and the inefficiency (\ref{size}) for different values of the parameter $T$ in the app selection probability (\ref{e_x}), respectively. We adopt the standard box plot\cite{box plot}, so that the black line and the purple line in the box represent the median and the mean, respectively; the edges of the boxes represent the $0.25$-quantile (Q1) and $0.75$-quantile (Q3); the whiskers span $1.5$ times the interquartile range beyond Q1 and Q3; and the circles are data points that fall outside the interval covered by the whiskers. The target coverage probability is $1-\alpha = 0.8$ and it is shown as a dashed line in Fig. \ref{fig:coverage_0}.

Fig. \ref{fig:coverage_0} demonstrates that the proposed CCKE guarantees the coverage condition (\ref{p1}) for all values of $T$, thereby validating Lemma \ref{lemma}. In contrast, CKE and NCCKE generally fail to achieve the target coverage $1-\alpha = 0.8$. The exception to this rule is NCCKE with $T\geq 100$. In fact, increasing $T$ in the selection probability $p(a|x)$ makes the selected app increasingly independent of the context $x$, thus reducing the covariate shift problem discussed in Section \ref{WCP}. Thus, in this case, NCCKE coincides with CCKE, meeting the coverage requirement (\ref{p1}). To see this, note that the weights in (\ref{weight}) tend to $1$ as $T\rightarrow \infty$, and thus the empirical distribution produced by NCCKE in (\ref{NCCKE distribution}) equals the distribution (\ref{weight distribution}) adopted by CCKE.


As seen in Fig. \ref{fig:size_0}, CKE undercovers the KPIs, providing insufficiently large prediction sets. In contrast, CCKE ensures the coverage condition (\ref{p1}) by properly increasing the prediction set size (\ref{w_interval_calculate}). Finally, NCCKE tends to have the same prediction set sizes as CCKE as $T$ increases, meeting the coverage requirement (\ref{p1}).

\begin{figure}[htbp]
    \centering
    \includegraphics[width=8cm]{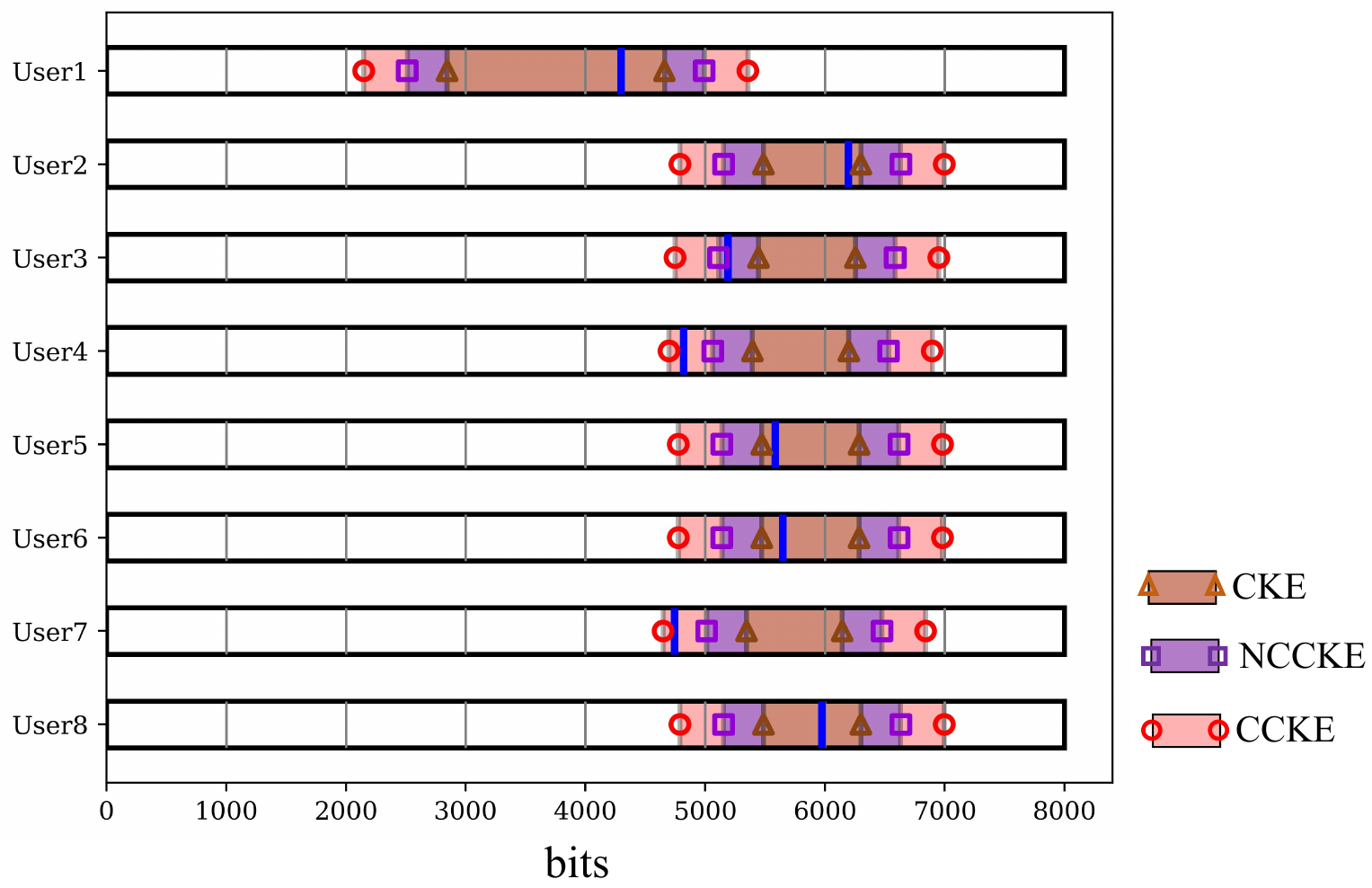}
    \vspace{-0.2cm}
    \caption{Prediction set for the remaining backlogs in the UEs' queues had app $a'=\mathrm{RR}$ been used when app $a=\mathrm{PFCA}$ was actually selected by the controller ($T=1$).}
    \label{fig:visual}
    \vspace{-0.2cm}
\end{figure}

To provide further insights into the performance of the schemes, Fig. \ref{fig:visual} visualizes the prediction intervals obtained by different methods for a specific realization of calibration and test data with $T=1$. The black rectangle represents the initial backlogs of each of the $K=8$ users, and the blue line represents the true remaining backlogs after running the target app $a'=\mathrm{RR}$. The prediction intervals are represented as shown in the legend. In line with the results in Fig. \ref{fig: coverage_size}, it is observed that the CKE and NCCKE output excessively wide intervals that cannot properly cover the true remaining backlogs for some of the users, while the CCKE can always cover the true remaining backlogs for all of the users.

To highlight the scalability of CCKE, Fig. \ref{fig: more_users} reports the coverage and inefficiency of different schemes with $K=8, 16,$ and $32$ users for $T=1$. As the number of users increases, the prediction model becomes less effective. As seen in Fig. \ref{fig:coverage_more_users}, this causes CKE to fail to satisfy the coverage requirements, while the proposed CCKE provides reliable prediction intervals regardless of the number of  $K$. As seen in Fig. \ref{fig:size_more_users}, this is done by increasing the size of the prediction interval to compensate for the performance degradation of the prediction model.

\begin{figure}[htp]
    \centering
    \vspace{-0.5cm}
    \subfigure[]{
    \includegraphics[width=7.5cm]{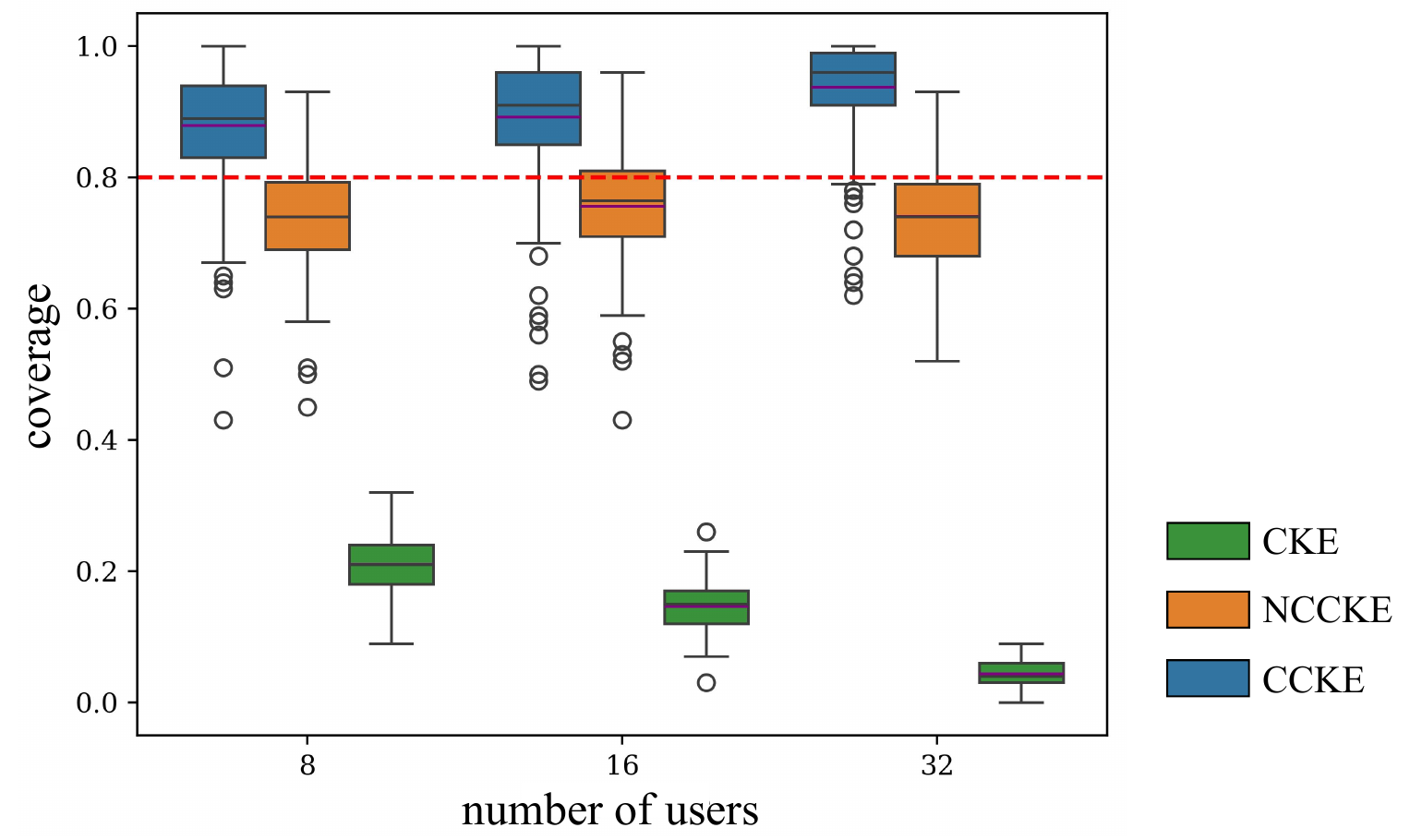}
    \vspace{-0.5cm}
    \label{fig:coverage_more_users}
    }
    \subfigure[]
    {
    \centering
    \includegraphics[width=7.5cm]{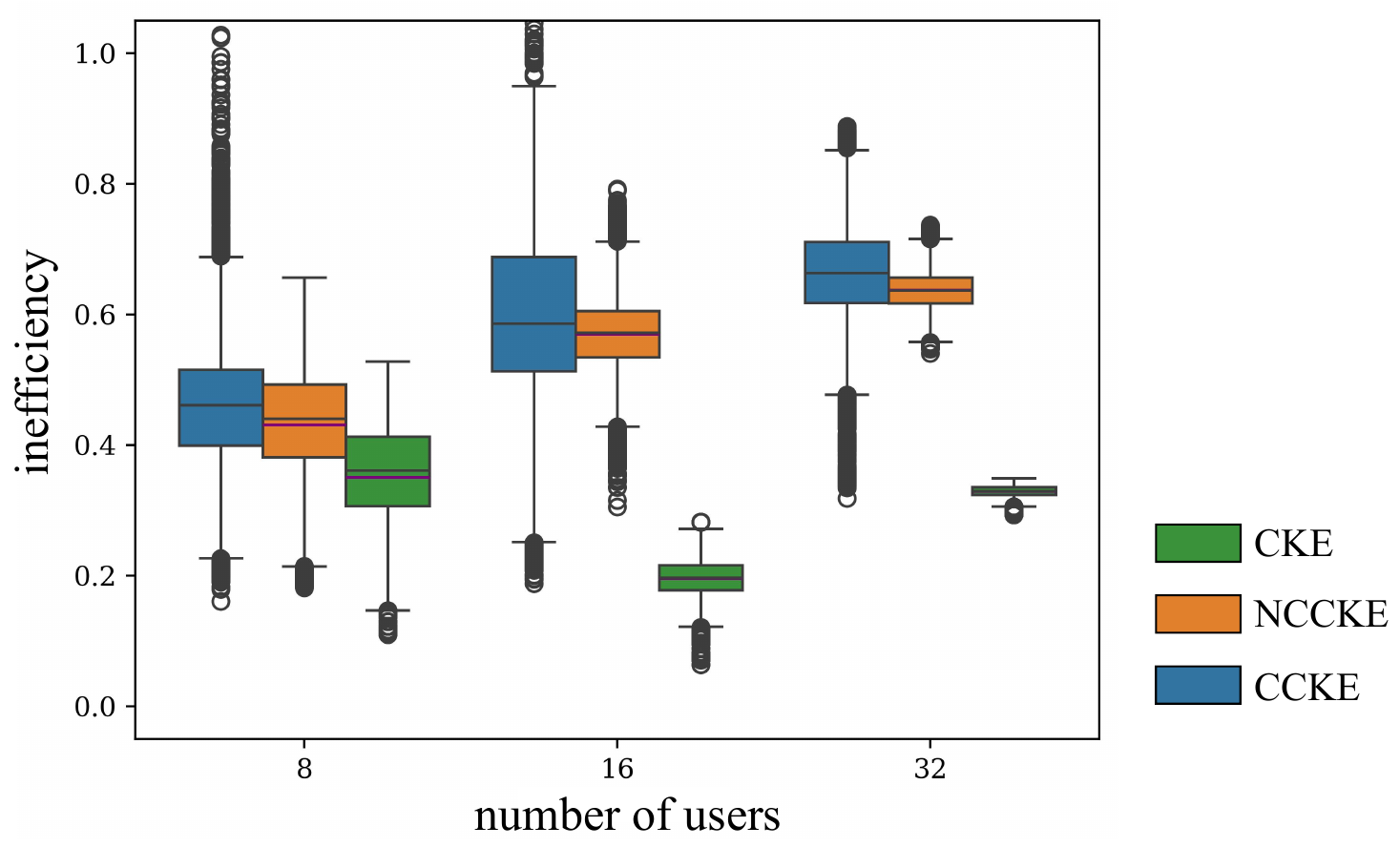}
    \label{fig:size_more_users}
    }
    \vspace{-0.2cm}
    \caption{Coverage and inefficiency performance for the counterfactual KPI analysis of the scheduling app RR given that the actual scheduling app is PFCA as a function of the number of users $K$ for $T=1$. The target coverage level is $1 - \alpha = 0.8$ (dashed line), and the results are averaged over $200$ experiments.}
    \label{fig: more_users}
\end{figure}

\section{Application to Physical-Layer Transmission}\label{example1}
In this section, we first provide further details on the multi-antenna transmission setting described in Section \ref{example_section}, which is then leveraged to demonstrate the use of CCKE in a wireless application at the physical layer.

\subsection{Setup}
Following Section \ref{example_section}, we consider a multi-antenna communication link in which the transmitter is equipped with $N_t=2$ antennas and the receiver has $N_r=2$ antennas. Prior to transmission, the transmitter estimates the average SNR, while also collecting information about the richness of the multipath channel. For example, using cameras, the transmitter may be able to estimate the number of scatters in the environment\cite{content}. 

We adopt a standard multipath channel model \cite{MIMO channel model}, according to which the $2\times 2$ channel matrix is given by
\begin{equation}\label{eq:H}
    H = \sqrt{\textrm{SNR}}\cdot\sum_{i=1}^{m}a_ie_r(\phi_{r,i})e_t(\phi_{t,i})^\dagger,
\end{equation}
where $a_i$ denotes the attenuation of path $i$; $m$ is the number of paths;  $\phi_{t,i}$ and $\phi_{r,i}$ denote the angles of  departure and arrival, respectively, for path $i$; and the steering vectors are given as 
\begin{equation}
   e_z(\phi) = \frac{1}{\sqrt{2}}
   \begin{bmatrix}
   1\\
   \exp(-j2\pi\Delta_z \cos(\phi)),
   \end{bmatrix}
\end{equation} 
where $z\in\{t,r\}$, and  $\Delta_z$ denotes the normalized separation between the transmit ($z=t$)  and receive ($z=r$) antennas. Accordingly, context $x$ includes the average SNR, denoted as  $\textrm{SNR}$, as well as the number of multipath components $m$, i.e., $x = (\textrm{SNR},m)$.

Based on the context information $x$, the transmitter chooses a transmission scheme $a$. A transmission scheme $a = (a_c,a_m)$ consists of the choice of a space-time coding method $a_c\in\mathcal{A}_c$ and of a constellation $a_m\in\mathcal{A}_m$. The set $\mathcal{A}_c$ may include different diversity-based space-time codes and multiplexing-based methods\cite{Alamouti code}, while the set $\mathcal{A}_m$ may include, e.g., BPSK and $M$-QAM for different values of integer $M$. Specifically, in this example, the set $\mathcal{A}_c$ encompasses Alamouti coding\cite{Alamouti code} and a basic multiplexing scheme transmitting a different symbol from each antenna, i.e., $\mathcal{A}_m = \{\mathrm{Alamouti}, \mathrm{multiplexing}\}$. Furthermore, the constellation set $\mathcal{A}_m$ includes BPSK and QPSK, i.e., $\mathcal{A}_m = \{\mathrm{BPSK},\mathrm{QPSK}\}$.

Finally, assuming a standard ARQ protocol, the KPI $y$ represents the retransmission latency measured in the number of transmission attempts. Denoting as $Y_{\rm{max}}$ the maximum allowed number of retransmissions, the latency KPI is limited in the range $1\leq y \leq Y_{\rm{max}}$.

The selection of the app $a$ at the controller is based on analytical approximations of the symbol error rate (SER) presented in \cite{SER on mul, SER on Ala}. Specifically, the conditional probability of choosing app $a$ given context $x$ is modeled as
\begin{equation}\label{e_x_example1}
    p(a|x) = \frac{\exp({1}/{\epsilon^a_xT})}{\sum_a\exp({1}/{\epsilon^a_xT})},
\end{equation}
where $\epsilon^a_x$ is an estimate of SER if transmission app $a$ was chosen under context $x$ and $T\textgreater 0$ is a temperature parameter playing the same role as in the previous example (see  Fig. \ref{fig:e_x}). 

The estimate $\epsilon^a_x$ is obtained by following the SER analysis on correlated Rayleigh channels in \cite{SER on mul, SER on Ala}, where the channel correlation properties are captured by the multipath parameter $m$ in (\ref{eq:H}), so that the off-diagonal elements of the channel covariance matrix equal $1/m$. This way, an increasing number of multipath components, $m$, causes a reduction in the spatial correlation\cite{correlation}. Specifically, we use the expressions in {\protect{\cite[Eq. 41]{SER on mul}}} and {\protect{\cite[Eq. 30]{SER on Ala}}}, respectively.

By (\ref{e_x_example1}), the selection of the higher modulation scheme QPSK and of the multiplexing transmission scheme becomes more likely as the corresponding SER estimates decrease, i.e., as the SNR and the number of paths increase. Furthermore, the parameter $T \textgreater 0$ in (\ref{e_x_example1}) makes it possible to control the dependence of the selected app on the context $x$. For $T\rightarrow 0$, the app selection becomes increasingly dependent on the context $x$. At the other extreme, for $T\rightarrow \infty$, the probability (\ref{e_x_example1}) tends to $0.25$, and hence all apps are selected with the same probability. The miscoverage level $\alpha$ is set as $0.2$, the maximum number of retransmissions is $Y_{\rm{max}} = 10$, and parameter $\epsilon$ in (\ref{size}) is set as $1$.

In the simulation, we generate $N^{\mathrm{tr}} = 3000$ data pairs for $\mathcal{D}^{\mathrm{tr}}$ and $N^{\mathrm{cal}} = 50$ data pairs for calibration data set $\mathcal{D}^{\mathrm{cal}}$. Specifically,  the $\textrm{SNR}$ for each context $x$ is randomly sampled from a truncated Gaussian distribution ranging from $-5$ dB to $15$ dB with the mean value as $5$ dB, and the multipath parameter $m$ in the channel model (\ref{eq:H}) is randomly and uniformly sampled from the integer interval $[1,10]$. The quantile regressor $\hat{q}_\tau(x|\phi)$ in (\ref{CQR}) is implemented via a $5$-layer feedforward neural network with input size $d_{\rm in}=2$, hidden sizes $d_{\rm hid}=[10, 10, 5]$, and output size $d_{\rm out}=2$.

\subsection{Performance Analysis}
In this subsection, we evaluate the coverage and inefficiency provided by CCKE in a scenario in which the actual app selected by the controller is multiplexing with a QPSK modulation, i.e., $a_c=\mathrm{multiplexing}$ and $a_m=\mathrm{QPSK}$. The goal is to predict the retransmission times that would have been observed if the Alamouti scheme with QPSK modulation had been selected instead.

\begin{figure}[htp]
    \centering
    \subfigure[Coverage performance (\ref{coverage})]{
    \includegraphics[width=7.6cm]{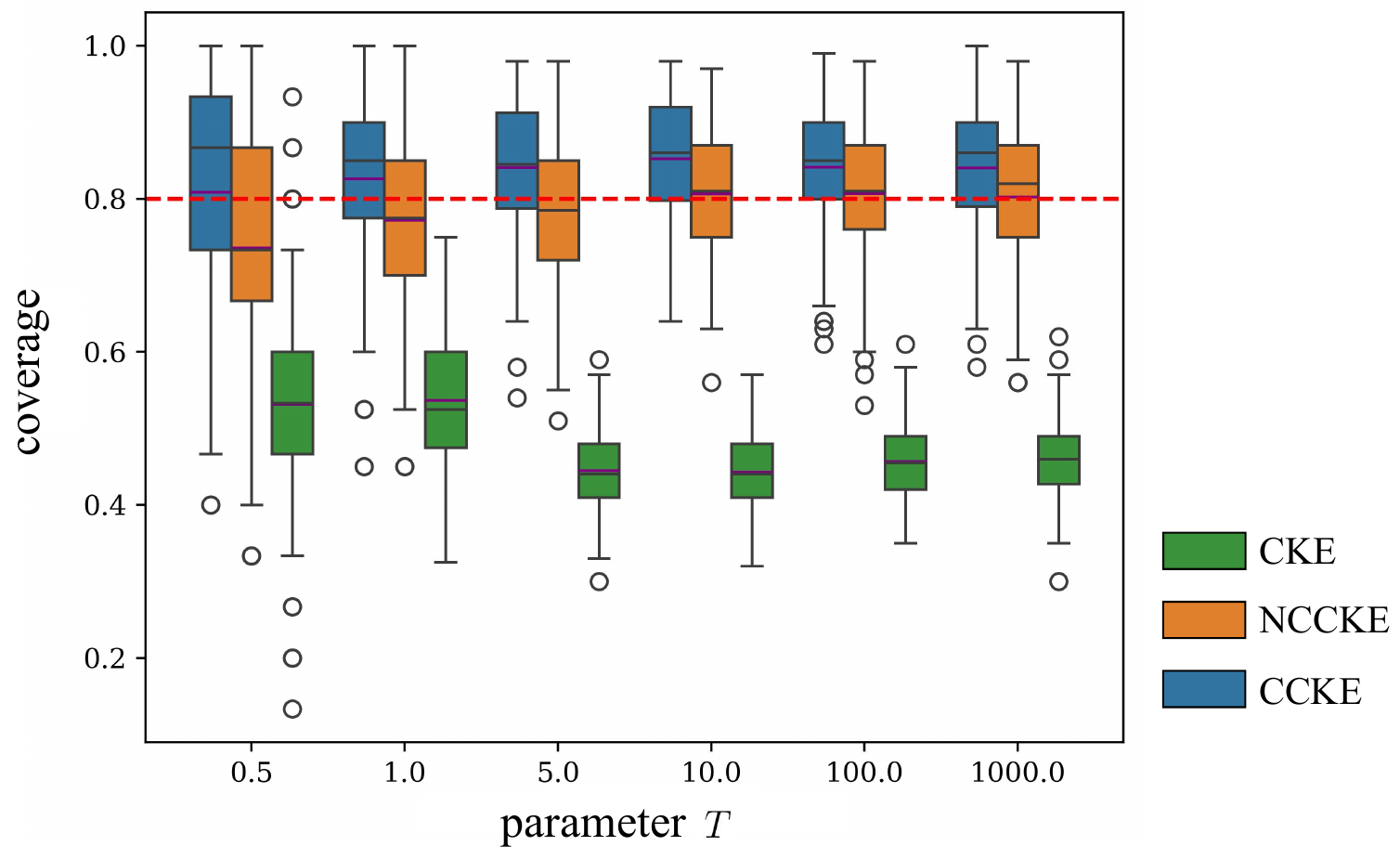}
    \vspace{-0.5cm}
    \label{fig:coverage_example1}
    }
    \subfigure[Inefficiency performance (\ref{size})]
    {
    \centering
    \includegraphics[width=7.5cm]{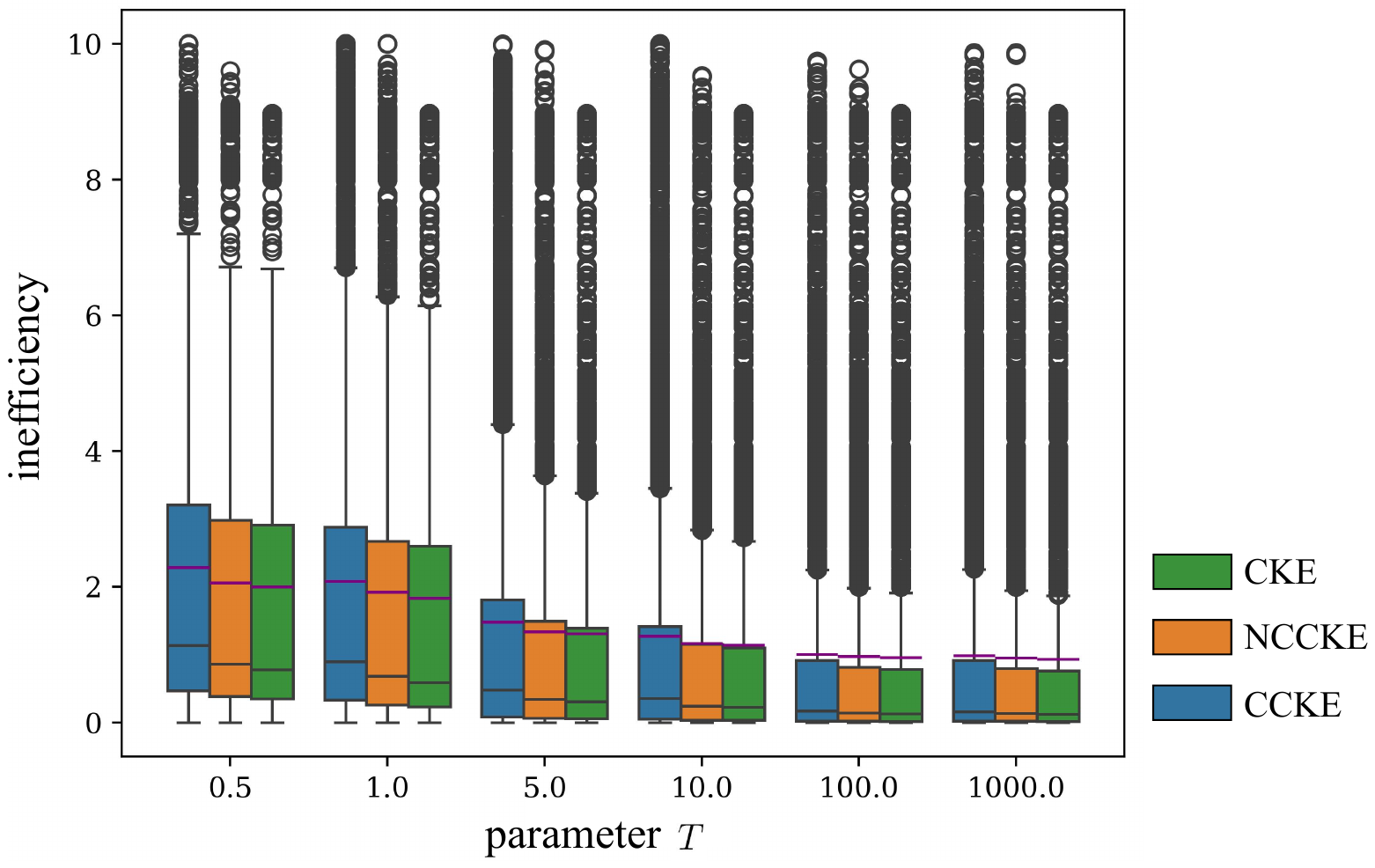}
    \label{fig:size_example1}
    }
    \caption{Coverage and inefficiency performance for the counterfactual KPI analysis of the transmission app based on Alamouti and QPSK given that the actual scheduling app uses multiplexing with QPSK as a function of parameter $T$ in (\ref{e_x_example1}). The target coverage level is $1-\alpha = 0.8$ (dashed line), and the results are averaged over $200$ experiments.}
    \label{fig: coverage_size_example1}
    \vspace{-0.3cm}
\end{figure}

As shown in Fig. \ref{fig:coverage_example1}, the proposed CCKE scheme always guarantees the coverage condition (\ref{p1}), while CKE and NCCKE generally fail to achieve the target coverage $1-\alpha = 0.8$ except for NCCKE with $T\geq 10$. With a growing parameter $T$ in the selection function (\ref{e_x_example1}), the selected app is increasingly independent of the context $x$. Thus, the covariate shift problem discussed in Section \ref{WCP} is gradually reduced, and NCCKE gradually coincides with CCKE, meeting the coverage requirement (\ref{p1}).

As seen in Fig. \ref{fig:size_example1}, CCKE ensures the coverage condition (\ref{p1}) by properly increasing the prediction set size. In this regard, NCCKE tends to have the same prediction set sizes as CCKE as $T$ increases, meeting the coverage requirement (\ref{p1}).

\section{Conclusions}\label{conclusion}
An important challenge in the design of next-generation wireless systems is to provide network operators with diagnostic, explainability, and optimization tools that can answer counterfactual queries about the KPIs that would have been obtained had a different choice been made by a controller. This work has made steps towards defining methods that can reliably address such ``what-if'' questions by leveraging the availability of an offline data set of logged tuples consisting of context, app identity, and KPIs. The proposed method, referred to as CCKE, leverages weighted conformal prediction to account for the covariate shift between logged data and test data, while producing provably valid ``error bars'' on the counterfactual estimates. We have showcased the operation of CCKE for two wireless applications operating at the medium access control layer and at the physical layer. 

Future work may focus on settings in which the app selection probability is only known implicitly through an algorithm, and thus it must be learned. Additionally, investigating the generalization of the proposed CCKE to distributed systems\cite{distributed CP1, distributed CP2}, possibly accounting also for temporal variability\cite{dynamic CP density ratio estimation1, dynamic CP density ratio estimation2, temporal CP 1, temporal CP 2}, are interesting directions. Furthermore, further research could incorporate digital twins to augment the logged data set for improved performance.


\section*{Appendix: Prediction Model for the Medium Access Control Problem} \label{App:MAC}

We treat the context $x$ in (\ref{context}) as a collection of $K$ two-dimensional tokens. This is formatted into a $2\times K$ matrix $x$, with each $k$-th column represents token $(b_k^{\mathrm{in}}, c_k)$. Let $W_Q\in\mathbb{R}^{d_h \times 2}$, $W_K\in\mathbb{R}^{d_h \times 2}$, $W_V\in\mathbb{R}^{d_o \times 2}$, $\hat{W}_Q\in\mathbb{R}^{d_h \times d_e}$, $\hat{W}\in\mathbb{R}^{d_h \times d_e}$, and $\hat{W}_V\in\mathbb{R}^{d_o \times d_e}$ be trainable parameters in the self-attention mechanism, where $d_h$, $d_o$, and $d_e$ are the hyperparameters that can be freely chosen. We first apply the self-attention mechanism following \cite{self-attention} as
\begin{equation}\label{MLP_input}
    x_{\mathrm{att}} = V\cdot\mathrm{softmax}\left(\frac{\mathcal{K}^TQ}{\sqrt{d_h}}\right),
\end{equation}
where $x_{\mathrm{att}}$ is a $d_o\times K$ matrix that contains the transformed tokens; $Q = W_Qx$, $\mathcal{K} = W_Kx$, $V = W_Vx$; and we have the row-wise softmax operation $\mathrm{softmax}([A]_{i,j}) = e^{[A]_{i,j}}/\sum_{j=1}^Je^{[A]_{i,j}}$ for $i=1,\dots, K$; $j=1,\dots, K$. Then, we apply a multi-layer feedforward neural network to each of the $K$ output tokens in (\ref{MLP_input}) to obtain the hidden embedding $x_e\in\mathbb{R}^{d_e\times K}$ as
\begin{equation}
    x_e = \mathrm{MLP}_1([x_{\mathrm{att}}]_k),
\end{equation}
where $[x_{\mathrm{att}}]_k$ is the $k$-th column of matrix $x_{\mathrm{att}}$ in (\ref{MLP_input}) for all $k=1,\dots, K$. By repeating the self-attention mechanism in (\ref{MLP_input}) for $x_e$ as
\begin{equation}\label{MLP_input_2}
    \hat{x}_{\mathrm{att}} = \hat{V}\cdot\mathrm{softmax}\left(\frac{\hat{\mathcal{K}}^T\hat{Q}}{\sqrt{d_h}}\right),
\end{equation}
where $\hat{Q} = \hat{W}_Qx_e$, $\hat{\mathcal{K}} = \hat{W}_Kx_e$, $\hat{V} = \hat{W}_Vx_e$, we input each of the $K$ output tokens in (\ref{MLP_input_2}) to another multi-layer feedforward neural network to obtain the estimated  $\alpha/2$- and ($1-\alpha/2$)-quantile of KPIs \{$b^{\mathrm{fin}}_k\}_{k=1}^K$ as
\begin{equation}
    \left(\hat{q}^k_{\alpha/2}, \hat{q}^k_{1-\alpha/2}\right) = {\mathrm{MLP}_2}([\hat{x}_{\mathrm{att}}]_k).
\end{equation}
Note that we apply the MLPs that are shared across all $K$ UEs in order to ensure permutation equivariance\cite{Deep Sets}. The simulation parameters of this prediction model are summarized in Table \ref{table3}.

\begin{table}[htpb]
    \renewcommand\arraystretch{1.2}
	\small
	\caption{Simulation Parameters}
	\label{table3}
	\centering
	\scalebox{1}{
	\begin{tabular}{c|c}
		\toprule
		\textbf{Parameter} & \textbf{Value}\\
		\hline\hline
        Dimension of attention $d_h$ & $10$\\
        \hline
        Dimension of output of attention $d_o$ & $10$\\
        \hline
        Dimension of hidden embedding $d_e$ & $10$\\
        \hline
        $\mathrm{MLP}_1$ & $\{10,10,10\}$\\
        \hline
        $\mathrm{MLP}_2$ & $\{10,10,2\}$\\
		\bottomrule
	\end{tabular}}
\end{table}

\end{document}